\documentclass[a4paper,12pt]{article}
\usepackage{amsmath,amssymb,amsfonts}
\usepackage{amsthm}
\usepackage[margin=1in]{geometry} 
\usepackage{graphicx,xcolor}
\usepackage{mathtools}
\usepackage{booktabs}
\usepackage{hyperref}
\hypersetup{urlcolor=blue, citecolor=green, linkcolor=blue}
\usepackage{caption}
\usepackage{tabularx}
\usepackage{subcaption}
\usepackage{placeins}
\usepackage[numbers,square,sort&compress,longnamesfirst]{natbib}
\usepackage{float} 

\usepackage{lmodern}
\usepackage[T1]{fontenc}
\usepackage{cellspace}
\newtheorem{lemma}{Lemma}
\newtheorem{prop}{Proposition}
\newtheorem{theorem}{Theorem}
\newtheorem{corollary}{Corollary}
\numberwithin{equation}{section}
\newtheorem{remark}{Remark}
\newtheorem{assumption}{Assumption}
\newtheorem{definition}{Definition}

\newtheorem*{sassumption}{Standing Assumption}

\numberwithin{theorem}{section}
\numberwithin{lemma}{section}
\numberwithin{prop}{section}
\numberwithin{corollary}{section}
\numberwithin{equation}{section}
\numberwithin{remark}{section}
\numberwithin{assumption}{section}
\numberwithin{definition}{section}
\numberwithin{example}{section}

\usepackage{scalerel,stackengine}

\newcommand\EE {\mathbb E}

\newcommand\RR {\mathbb R}

\newcommand{\nihil}[1]{}

\allowdisplaybreaks

\begin{document}
\title{Infinite Horizon Optimal Consumption: Intertemporal Hedging under Epstein-Zin Preferences}
\author{
Erhan Bayraktar\thanks{University of Michigan, Department of Mathematics, Ann Arbor MI 48109, United States. 
\texttt{erhan@umich.edu}. Supported in part by the National Science Foundation under grant DMS-2507940 and the Susan M. Smith Chair.}
\and
Emmet Lawless\thanks{University of Michigan, Department of Mathematics, Ann Arbor MI 48109, United States. \texttt{elawless@umich.edu}.}
}
\date{}
\maketitle


\begin{abstract}
We study an infinite-horizon optimal consumption-investment problem for an investor with Epstein-Zin stochastic differential utility in an incomplete market with stochastic investment opportunities. Risk aversion and intertemporal substitution are separated, and we work in the regime $\theta\in(0,1)$, where there exists a unique generalised utility process for arbitrary non-negative progressively measurable consumption streams. Our main contribution is a variational characterisation of the value function. We show that the value function is the unique minimiser of a functional whose Euler-Lagrange equation coincides with the Hamilton-Jacobi-Bellman equation. Although the functional may be non-convex, the direct method yields existence, and we prove that every minimiser is a strictly positive, bounded classical solution. A verification theorem identifies any minimiser with the value function and gives feedback representations for optimal consumption and investment policies. The proof combines a change of measure to the myopic probability with uniqueness results for Epstein-Zin BSDEs and a perturbation argument for optimality. Examples with stochastic volatility, Gaussian excess returns, and fat-tailed excess returns illustrate the scope of the framework and its implications for intertemporal hedging.
\end{abstract}

\phantom{\textbf{Keywords:}} infinite horizon, stochastic investment opportunities, calculus of variations
\smallskip
	
\smallskip
	
	
\newpage

\section{Introduction}
The classical literature on optimal consumption in incomplete markets has focused mainly on time-additive utility, which forces preferences to be captured by a single risk-aversion parameter (denoted $\gamma$). Epstein-Zin stochastic differential utility (SDU) offers a framework that separates two competing forces: risk aversion and attitudes toward intertemporal substitution (captured by the elasticity of intertemporal substitution (EIS), denoted $1/\delta$).\footnote{Other works in the literature use $\psi$ to denote EIS; in our notation, $\psi=1/\delta$.} The latter allows direct modelling of an agent's preference over the timing of uncertainty resolution, as captured by the relative ordering of $\gamma$ and $\delta$. If $\delta < \gamma$, agents prefer early resolution, while if $\gamma < \delta$, agents prefer late resolution. If $\delta=\gamma$, the agent is indifferent to the timing of uncertainty resolution, which gives the special case of time-additive preferences. In this paper we work in the regime $\theta \coloneqq \frac{1-\gamma}{1-\delta} \in (0,1)$, which accommodates both preferences regarding the timing of uncertainty resolution.

The core focus of this paper is to extract economic insight into the optimal behaviour of an infinite-horizon investor with Epstein-Zin preferences and stochastic investment opportunities. We employ the variational approach to portfolio choice introduced in \citet{GLT:VariationalApproach(2025)} to characterise the value function as the unique minimiser of a \emph{non-convex} functional whose Euler-Lagrange equation is exactly the HJB equation.\footnote{A variational formulation of the classical finite-horizon Merton portfolio optimisation problem is presented in the short note of \citet{lorig(2025)}.}

First, we construct a solution to the variational problem via the so-called direct method, yielding the existence of a minimiser. Next, we prove \emph{any minimiser} is a bounded classical solution to the HJB equation by employing truncation arguments, Sobolev embeddings, and elliptic regularity results.

Our verification argument proves that an arbitrary minimiser is indeed the value function; the minimiser is therefore unique despite non-convexity. Verification requires two key steps. Under a change of measure to the myopic probability measure (see \citet{GR:Portfolios&RiskPremia(2012)}), we show that the candidate value function solves the Epstein-Zin BSDE for the candidate optimal controls. Uniqueness of the utility process induced by the candidate controls then follows from the general market-free analysis in \citet{HHJ(2023):EZutilityII}. To prove optimality, we extend the perturbation approach to verification used in \citet{HHJ(2023):EZutilityII} for a Black-Scholes-Merton market to our setting with a stochastic state variable.

The appeal of the variational approach is that the main object under study becomes a functional rather than the HJB equation---a second-order semilinear (and possibly singular) ODE with no boundary conditions. The abstract tools of weak convergence and Sobolev embeddings allow us to prove strong results under relatively weak assumptions on the problem data. This approach is distinct from the direct analysis of the HJB equation in \citet{GHH:ConsumptionINvestmentSFM(2025)}, where the authors develop a general theory of second-order semilinear elliptic ODEs on open domains for the optimal consumption problem with CRRA preferences. The variational characterisation also provides a natural object to discretise, removing the need to impose economically unmotivated boundary conditions when solving the ODE directly.

Studying elliptic equations via a variational formulation is commonplace in the PDE literature, see \citep{badiale2010semilinear,Struwe:(2008)}. To the best of our knowledge this approach has not been exploited in the mathematical finance literature on optimal consumption prior to the work of \citet{GLT:VariationalApproach(2025)}.

To emphasise the utility of our results, we focus on examples featuring stochastic volatility and stochastic excess returns. It is now widely accepted that the assumption that excess returns have Gaussian tails is not supported by the data. However, this fact is routinely ignored in the optimal consumption literature because departures from standard Gaussian-type models typically destroy analytical tractability. Our framework addresses this shortcoming and can be used for a wide class of non-standard models. We compare optimal policies in the standard Gaussian regime with those under fat-tailed excess returns. We find that fat-tailed returns generate a convex-to-concave transition in the optimal consumption ratio as a function of the state variable, a feature absent under Gaussian returns.


The rest of this article is organised as follows. Section \ref{sec. Literature Review} reviews relevant literature. In Section \ref{sec. Model and problem formulation} we formulate the agent's optimisation problem, state all assumptions, and introduce the variational problem. Section \ref{sec. Main result} contains our main result and explains the connection between the HJB equation and the variational problem. In Section \ref{sec. Examples} we present three examples to which our main theorem can be applied. The verification argument is presented in Section \ref{sec. Verification}, we give concluding remarks in Section \ref{sec. Conclusion}, and all proofs are collected in the Appendix.

\section{Literature Review} \label{sec. Literature Review}

Recursive utility of Epstein-Zin type was introduced in the seminal papers of Epstein and Zin \citep{EZ:(1989),EZ:(1991)}, where the authors extended the work of \citet{KP:Econometrica(1978)} on dynamic choice over lotteries to a discrete-time infinite-horizon setting. Extensions to continuous time for general aggregators, leading to the notion of stochastic differential utility (SDU), were developed by Duffie and Epstein \citep{DE(1992):Econometrica(SDU),DE(1992):RFS}. Subsequent work connected recursive utility to BSDE, PDE, martingale, and utility-gradient methods \citep{DL(1992):JME,DS(1994):JME,SS(1999):JET-ConsumptionSDU,CV(1999):QJE,SS(2003):SPA}. These papers provide the basic analytic language for the present problem, but much of the early theory either assumes Lipschitz aggregators, complete markets, unit elasticity of intertemporal substitution, or model structures that allow explicit solutions.




Rigorous verification arguments were provided in \citet{KSS:FS(2013)EZ} in the finite- and infinite-horizon cases for the Epstein-Zin aggregator and non-unit elasticity of intertemporal substitution. The authors relied on a one-sided Lipschitz condition and restrictions on admissible parameters\footnote{The authors required the following identity to hold: $\frac{1}{\delta}=2-\gamma+\frac{(1-\gamma)^2}{\gamma}\rho^2$. Here $\rho$ is the correlation between the Brownian motion driving the risky asset and the Brownian motion driving the state variable.} to linearise the associated HJB equation. Despite these advances, verification was still ad hoc and proven on a model-by-model basis. For the infinite-horizon case, the assumptions imposed in \citet{KSS:FS(2013)EZ} are difficult to check and rely heavily on the explicit representation of the solution of the linearised HJB equation. Following this work, the link between discrete-time recursive preferences and SDU was rigorously established via convergence arguments in \citet{KS(2014):JET-EZdiscreteContinuous}.

On finite horizons, later work removed some of the earlier parameter
restrictions through fixed-point, BSDE, and duality methods, typically under boundedness, Lipschitz, or local-Lipschitz assumptions on model coefficients \citep{KSS:FS(2017)EZ,Xing:EZConsumption(2017),MX(2018):MF-EZ}. A now substantial finite-horizon literature treats constraints, random horizons, mean-field models, leverage restrictions, and non-Markovian
models \citep{AH(2021):SIFIN-EZ,AH(2023):MF-EZ,FT(2023):PUQR-EZ-strat-constraints,HLT:(2024)EZConstrained,FH(2025):EZMeanField,TTZ(2025)EZcosumptionConstraint,TTZZ(2025):EZ-leverage,FTZ(2026):SPA-EZ-NonMarkov}. These results show that the finite-horizon problem is comparatively well developed,
but they do not provide a general infinite-horizon verification theorem for incomplete markets with stochastic investment opportunities.



The infinite-horizon literature is far less developed. Existing results work in a Black-Scholes-Merton market \citep{MMKS:MF(2020)EZTransactionCosts}, in the regime $\theta<0$ with bounded model coefficients \citep{Dang(2021):Thesis,Shigeta(2026):FS-EZ}, or in the special case of unit EIS, where explicit solutions are available \citep{KGH(2025):CNSNS-EZ}. The works of Herdegen et al. \citep{HHJ(2023):EZutilityI,HHJ(2023):EZutilityII,HHT(2024):EZ-TC,HHJ(2025)FS:EZUtilityProperSol} give a comprehensive account of existence and uniqueness of utility processes for the Epstein-Zin aggregator; however, verification is only treated in a Black-Scholes-Merton market.

To the best of our knowledge, the literature does not contain a complete treatment of the infinite-horizon Epstein-Zin optimal consumption-investment problem in incomplete markets without artificial restrictions on model coefficients or preference parameters. The current work fills that gap in the case $\theta \in (0,1)$. Methodologically, our paper is closest in spirit to the variational
approach of \citet{GLT:VariationalApproach(2025)} for CRRA consumption-investment problems. The Epstein–Zin case changes the problem in two ways: the variational functional becomes non-convex in the relevant regime, and verification must identify a recursive utility process rather than a CRRA value function. Our contribution to the literature is threefold:
\begin{itemize}
    \item[(i)] We provide an existence, uniqueness, regularity, and verification theorem for the Epstein-Zin optimal consumption-investment problem over an infinite horizon with any number of risky assets and a scalar state variable. This theorem gives a general verification result that requires no bespoke model-by-model analysis.
    \item[(ii)] We remove restrictive pointwise assumptions on model coefficients. We impose H\"older regularity and mild joint-integrability conditions on model coefficients, together with positive myopic consumption. These assumptions allow for highly non-linear and non-affine market models.  
    \item[(iii)] We provide a variational characterisation of the value function, yielding a simple numerical scheme that can be used for comparative statics. In addition, the variational characterisation removes the need to solve the HJB equation directly, thereby avoiding technical analysis of second-order semilinear elliptic equations on open domains with no initial conditions.
\end{itemize}

\section{Model} \label{sec. Model and problem formulation}

This section details the market model, the agent's consumption-investment problem, and the core assumptions used throughout the article.
\subsection{Market}
The market is modelled as a multivariate autonomous diffusion in which the risk-free rate, excess return, and volatility depend on an exogenous scalar state variable $(Y_t)_{t\ge 0}$ that cannot be traded. The market dynamics are
\begin{align}
\frac{d S_t^0}{S_t^0} & = r(Y_t) \, dt, \label{eq. Risk-Free asset}\\
 \frac{d S_t^i}{S_t^i} & = r(Y_t) \, dt + dR_t^i, \quad \quad &1 \leq i \leq n, \label{eq. S dynamics}\\  
     dR_t^i &= \mu_i(Y_t) \, dt + \sum_{j=1}^n \sigma_{ij}(Y_t) \, dZ_t^j, \quad \quad &1 \leq i \leq n,\label{eq. R dynamics} \\
dY_t &= b(Y_t) \, dt + a(Y_t) \, dW_t, {\color{white}\sum} \label{eq. Y dynamics}
\end{align}
where $Z=(Z^1,\ldots,Z^n)$ and $W$ are multivariate and scalar Brownian motions, respectively, and $(R_t^i)_{t\geq 0}$ is the cumulative excess-return process for the $i^{th}$ risky asset. The open connected set $E\subseteq\mathbb{R}$ denotes the range of the state variable $Y$. The correlation vector between $Z$ and $W$ is denoted $\rho(Y_t)$, yielding the covariation process
\begin{equation}
    d\langle Z^i,W \rangle_t = \rho_i(Y_t)dt,  \quad  1 \leq i \leq n. \label{eq. Covariation dynamics (R,Y)}
\end{equation}
We impose the condition
\begin{equation} \label{eq. correlation}
    \rho(Y_t)^{\top}\rho(Y_t)=c \in [0,1]
\end{equation}
which, in general, allows the market to be incomplete. To ease notation, the dependence of coefficients on the state variable $Y_t$ is henceforth omitted unless ambiguity arises.
Denote the covariation matrix of excess returns $(R_t^i)_{1\leq i \leq n}$ by $\Sigma = \sigma \sigma^{\top}$ and set $\Upsilon = \sigma \rho a$ as the $n \times 1$ covariation matrix between the risky assets and the state variable.

We denote by $(\Omega,\mathcal{F},\mathcal{F}_t,\mathbb{P}^y)$ the underlying probability space satisfying the usual conditions\footnote{The superscript $y$ is used to emphasise that the construction of the process $(R,Y)$ depends only on the process $Y$.}. In Section \ref{subsec. Assumptions and Notation}, we introduce assumptions on the model coefficients that allow rigorous construction of the market.

\subsection{Consumption-Investment Problem} \label{subsec. Consumption-Investment Problem}

An investor with initial wealth $x$ dynamically chooses consumption and investment. Denote by $(\pi_t)_{t\geq 0}$ the vector of proportions of wealth invested in the risky assets and by $(l_t)_{t\geq 0}$ the consumption-to-wealth ratio (i.e., the agent's consumption rate is $c_t=X_t \,l_t$). The budget equation is
\begin{align} \label{eq. Wealth Process Dynamics}
    \frac{d X_t^{\pi,l}}{X_t^{\pi,l}} & = r \, dt + \pi_t^\top \, dR_t - l_t \, dt .
\end{align}
For a fixed consumption stream $(c_t)_{t \geq 0}$, the associated utility process is a solution to the equation:

\begin{equation} \label{eq. discounted Utility Process}
    V_t^l=\EE\left[\int_t^{\infty}e^{-\beta (s-t)}\frac{(X_s\,l_s)^{1-\delta}}{1-\delta}\left((1-\gamma)V^l_s\right)^{1-\frac{1}{\theta}}ds \,\Big |\, \mathcal{F}_t \right].
\end{equation}

Here $\gamma \in \mathbb{R}^+ \setminus \{1\}$ denotes the agent's risk aversion, $\delta \in \mathbb{R}^+ \setminus \{1\}$ denotes the agent's intertemporal substitution aversion, and $\theta \coloneqq \frac{1-\gamma}{1-\delta}$. As shown in \citet[Remark 2.3]{HHJ(2023):EZutilityII}, the discount factor $e^{-\beta t}$ can be eliminated via a change of numéraire. In particular, the optimal consumption problem with discounting is equivalent to the optimal consumption problem without discounting but with a modified risk-free rate $\tilde{r}=r-\frac{\beta}{1-\delta}$.\footnote{In the papers of Herdegen et al. \citep{HHJ(2023):EZutilityI,HHJ(2023):EZutilityII}, the authors use the convention that $\mu$ represents the drift of the risky assets, while in the current work $\mu$ is the \emph{excess return} above the risk-free rate and hence does not need to be modified.} For simplicity, we henceforth set $\beta=0$ without loss of generality. With this convention, the equation for the utility process generated by a policy pair $(\pi,l)$ is
\begin{equation} \label{eq. Utility Process}
    V_t^{\pi,l}=\EE\left[\int_t^{\infty}\frac{(X_s^{\pi,l}\,l_s)^{1-\delta}}{1-\delta}\left((1-\gamma)V^{\pi,l}_s\right)^{1-\frac{1}{\theta}}ds \,\Big |\, \mathcal{F}_t \right].
\end{equation}
This formulation corresponds to the Epstein-Zin aggregator defined by
\begin{equation} \label{eq. EZ aggregator}
    f(c,v)\coloneqq\frac{c^{1-\delta}}{1-\delta}\left((1-\gamma)v\right)^{1-\frac{1}{\theta}}.
\end{equation}
This discounted formulation was introduced in \citet{HHJ(2023):EZutilityI} as an alternative to the classical difference-form aggregator, which appears in essentially all earlier works; see, for instance, \citep{KSS:FS(2013)EZ,KSS:FS(2017)EZ,Xing:EZConsumption(2017)}. When $\theta>0$, any utility process that solves the discounted Epstein-Zin BSDE also solves the corresponding difference-form BSDE. This equivalence is no longer true when $\theta<0$ and is discussed in detail in \citep{HHJ(2023):EZutilityI,Shigeta(2026):FS-EZ}.

For a fixed pair of policies $(\pi,l)$, existence and uniqueness of a solution to \eqref{eq. Utility Process} depend heavily on the parameters $\gamma$ and $\delta$. With this in mind, we introduce a standing assumption that defines the parameter regime used throughout the paper.
\begin{sassumption} \label{assum. Parameter regime} 
    $\theta \in (0,1)$.
    \nihil{
    $\left(1-\left(1-\frac{1}{\gamma}\right)\rho^{\top}\rho\right)\delta\theta\leq 1$, and $\left(1-\left(1-\frac{1}{\gamma}\right)\rho^{\top}\rho\right) \neq 0$. 
    }
\end{sassumption}
\nihil{
It was observed in \citet{HHJ(2023):EZutilityI} that the classical formulation and primal approach to the Epstein-Zin optimal consumption problem is economically ill-founded when $\theta<0$. The authors construct utility processes exhibiting bubble-like behaviour, where the agent earns utility not from immediate consumption, but from anticipated future consumption that is never realised. We provide a rigorous argument supporting this claim by illustrating that any bounded solution to the associated HJB equation \emph{does not} identify the value function when $\theta<0$. This finding provides further support for the alternative formulation proposed in \citet{Shigeta(2026):FS-EZ}, which implements a monotone transformation of the utility process to avoid bubble-like behaviour when $\theta<0$.
}

The assumption $\theta \in (0,1)$ makes the general market-free analysis in \citet{HHJ(2023):EZutilityII} applicable, ensuring existence and uniqueness of a generalised utility process \citep{HHJ(2023):EZutilityII} for any non-negative progressively measurable consumption stream. Indeed, if $\theta>1$, solutions to \eqref{eq. Utility Process} are no longer unique. Non-uniqueness when $\theta>1$ was studied extensively by \citet{HHJ(2025)FS:EZUtilityProperSol}, where the authors introduced the notion of ``\emph{proper}'' utility processes; the general theory was then applied to solve the optimal consumption problem in a Black-Scholes-Merton market. The case $\theta<0$ was studied in \citep{Dang(2021):Thesis,Shigeta(2026):FS-EZ}.

Next we introduce the class of admissible strategies.

\begin{definition} \label{def. admissable policy}
    The set of admissible policies, denoted by $\mathcal{A}$, consists of $(\mathcal{F}_t)_{t\geq0}$-adapted processes $(\pi,l)$ such that:
    \begin{enumerate}
        \item $\pi$ is integrable with respect to $R$.
        \item $l_t\geq0$ almost surely for $t\geq0$.
        \item There exists a unique strong solution $X^{\pi,l}$ to \eqref{eq. Wealth Process Dynamics}.
        \item $c_t=X_t^{\pi,l}l_t$ is progressively measurable.
    \end{enumerate}
\end{definition}
The dynamics \eqref{eq. Wealth Process Dynamics} automatically ensure $X_t^{\pi,l}\geq0$ almost surely; hence non-negativity of the associated consumption stream is guaranteed. Thus, for any policy pair $(\pi,l)\in  \mathcal{A}$, there exists a unique generalised utility process $V^{\pi,l}$ satisfying \eqref{eq. Utility Process}. The agent's problem is to find
\begin{equation} \label{eq. agent's goal}
    \hat{V}(x,y) =\sup_{(\pi,l)\in \mathcal{A}}V_0^{\pi,l}.
\end{equation}

With the optimisation problem now formulated, we introduce some mild assumptions that are used in all subsequent analysis.

\subsection{Assumptions and Notation} \label{subsec. Assumptions and Notation}
 $C^{k,\alpha}(E,\mathbb{R}^{n\times m})$ denotes the space of $\mathbb{R}^{n\times m}$-valued functions on $E$ that are $k$ times continuously differentiable and whose $k^{th}$-order partial derivatives are $\alpha$-H\"older continuous for some $\alpha \in (0,1)$.
\begin{assumption} \label{assum. Market model is well posed} \textup{(Well-posedness condition)} We assume the following:
\begin{itemize}
    \item[(i)]    For some $\alpha \in (0,1)$, $r \in C^{1,\alpha}(E;\mathbb{R})$, $\mu \in C^{1,\alpha}(E;\mathbb{R}^n)$, $b \in C^{1,\alpha}(E;\mathbb{R})$, $a^2 \in C^{2,\alpha}(E;\mathbb{R})$, $\Sigma \in C^{2,\alpha}(E;\mathbb{R}^{n \times n})$ and $\Upsilon \in C^{2,\alpha}(E;\mathbb{R}^n)$. For all $y \in E$, $\Sigma(y)$ is strictly positive definite and $a^2(y)>0$.
    \item[(ii)] There exists a unique solution to the martingale problem on $\mathbb{R}^{n} \times E$ associated with
            \begin{align*}
\hat{L} &= \frac{1}{2} \sum_{i=1}^{n+1}\sum_{j=1}^{n+1} \hat{A}_{i,j}(\cdot) \frac{\partial^2}{\partial x_i \partial x_j} + \sum_{i=1}^{n+1} \hat{b}_i(\cdot) \frac{\partial}{\partial x_i} 
\end{align*} 
where $\hat{A} = \begin{pmatrix} 
\Sigma & \Upsilon \\
\Upsilon^\top & a^2
\end{pmatrix}$ and
$\hat{b} = \begin{pmatrix} 
\mu\\
b
\end{pmatrix}$.
\end{itemize}
\end{assumption}

Assumption \ref{assum. Market model is well posed} is sufficient to construct an $\mathbb{R}^n \times E$-valued process $(R,Y)$ with continuous trajectories that satisfies the market dynamics \eqref{eq. R dynamics}, \eqref{eq. Y dynamics} and \eqref{eq. Covariation dynamics (R,Y)}. The rigorous construction is outlined in detail in \citep{GR:Portfolios&RiskPremia(2012),GW:ConsumptionIncompleteMarkets(2020),GLT:VariationalApproach(2025)}.

We rule out arbitrage by assuming the existence of a martingale measure:

\begin{assumption} \label{assum. Market is arbitrage free}\textup{(Arbitrage-free market)} There exists a probability measure $\mathbb{Q}^y$ such that $\mathbb{Q}^y|_{\mathcal{F}_t}$ and $\mathbb{P}^y|_{\mathcal{F}_t}$  are equivalent for every $t \in \mathbb{R}_+$ and $S/S^0$ is a $\mathbb{Q}^y$-local martingale.
\end{assumption}

 Next, we impose the same assumptions used in \citet{GLT:VariationalApproach(2025)}.

\begin{assumption} \label{assum. kappa is strictly positive}  \textup{(Positive myopic consumption)}
\noindent $\inf_{y \in E}\kappa (y)>0$, where
\begin{equation} \label{eq. Merton consumption kappa}
     E \ni  y \mapsto \kappa(y)=\left(1-\frac{1}{\delta}\right)\left(r(y) + \frac{\mu(y)^\top\Sigma^{-1}(y)\mu(y)}{2\gamma}\right).
\end{equation}
\end{assumption}
The function $\kappa$ is the optimal consumption ratio if investment opportunities were constant (i.e. if the state variable $Y_t$ were constant). Henceforth we refer to $\kappa$ as the \emph{myopic} consumption ratio. In a market with constant investment opportunities, strict positivity of $\kappa$ is a necessary condition for well-posedness; see \citep{HHJ(2023):EZutilityII}.

Next, we assume the state variable admits a stationary density (denoted $\eta$) and is recurrent under an equivalent probability measure. The dynamics under this probability measure are
\begin{equation} \label{eq. shifted state variable}
    d\tilde{Y}_t=\tilde{b}(\tilde{Y}_t)dt+a(\tilde{Y}_t)d\tilde{W}_t
\end{equation}
where
\begin{equation} \label{eq. tilde b}
    \tilde{b} : E \to \mathbb{R},\, y \mapsto  \,b(y)- \left(1-\frac{1}{\gamma}\right)\left(\Upsilon(y)^\top\Sigma^{-1}(y)\mu(y)\right).
\end{equation}

\begin{assumption}(Long-run stationarity) \label{assum. eta is a density} For any $x_0 \in E$, the mapping
\begin{equation} \label{eq. 1-d eta}
    \eta \colon E \to \mathbb{R}_+ ,\, y  \mapsto \frac{1}{ a(y)^2}\exp{\left(2\int_{x_0}^y \frac{\tilde{b}(u)}{\, a(u)^2}du\right)},
\end{equation}
is integrable and hence normalises to a probability density on $E$.
\end{assumption}
Recall that, for an arbitrary $x_0 \in E$,
\begin{equation} \label{eq. scale function of shifted state variable}
   E \ni  y \mapsto p(y)\coloneqq\int_{x_0}^y \exp\left(-2\int_{x_0}^u \frac{\tilde{b}(s)}{a(s)^2}ds\right)du,
\end{equation}
is the scale function of \eqref{eq. shifted state variable}. 
\begin{assumption}[Recurrence] \label{assum. Recurrence}
The scale function $p$ of \eqref{eq. shifted state variable}, defined in \eqref{eq. scale function of shifted state variable}, satisfies
    \begin{equation*}
        \lim_{y \to \partial E_+}p(y)=+\infty,
        \qquad
        \lim_{y \to \partial E_-}p(y)=-\infty;
    \end{equation*}
    where $\partial E_+,\partial E_{-}$ are the right and left endpoints of $E$, respectively.
\end{assumption}

\begin{remark}
    We do not assume that the state variable admits a stationary density under the physical probability measure $\mathbb{P}^y$. We only require the existence of an invariant density under the distorted measure obtained by shifting the drift \eqref{eq. tilde b}.
\end{remark}

\begin{remark}
    We emphasise that Assumptions \ref{assum. Market model is well posed}-\ref{assum. Recurrence} are exactly the same as those in \citet{GLT:VariationalApproach(2025)}, where the authors study an infinite-horizon optimal consumption problem under CRRA preferences. Hence, the results in this article generalise and extend the arguments in \citep{GLT:VariationalApproach(2025)} to cover Epstein-Zin preferences in incomplete markets without additional assumptions.

    When the market is incomplete, Theorem 5.2 in \citet{GLT:VariationalApproach(2025)} is restricted to the case $\gamma \in (0,1)$. With minor modifications to the proofs, one can extend this result to the regime $\gamma>1$; this extension is contained as a special case in the present article and is obtained by setting $\theta=1$ (and hence $\delta=\gamma$).
    
\end{remark}
\subsection{Variational Problem} \label{subsec. Variational Problem}
We follow the approach of \citet{GLT:VariationalApproach(2025)} and characterise the value function as the solution of a variational problem. Consider the functional

\begin{equation} \label{eq. Functional I}
    I: g \mapsto I(g) \coloneqq \int_E \left(\frac{\nu\, a^2(y)(g'(y))^2}{2} +\kappa(y) g(y)^2 -\frac{2g(y)^{2-\nu}}{(2-\nu)}\right)\eta(y)dy,
\end{equation}
where 
\begin{equation} \label{eq. nu power non-linearity}
    \nu \coloneqq \frac{\gamma}{(\gamma+(1-\gamma)\rho^{\top}\rho)\delta\theta}.
\end{equation}
\begin{remark}
    $\theta \in (0,1)$ implies $\nu >0$.
\end{remark}
We minimise \eqref{eq. Functional I} in the positive cone of the weighted Sobolev space
\begin{equation} \label{def. weighted sobolev space}
     H(E;\RR) \coloneqq \left\{u \in L^1_{loc}(E;\mathbb{R}): \int_E |u(y)|^2\eta(y) dy+\int_E |u'(y)|^2 a^2(y) \eta(y) dy < +\infty\right\}.
\end{equation}
 Denote the positive cone of $H(E)$ by $H(E)_+=\{u\in H(E): u \geq0 \,\, a.e.\}$. With this notation, the variational problem is
\begin{equation} \label{eq. 1-d minisation problem}
   \inf_{g \in H(E)_+}  I(g).
\end{equation}
\begin{theorem}\label{thm. Existence uniqueness regularity and boundedness solution to variational problem}
There exists a solution $g\in C^2(E;\mathbb{R})$ to the minimisation problem \eqref{eq. 1-d minisation problem} satisfying
$0<g\leq M$ on $E$ for some $M>0$.
\end{theorem}

Theorem \ref{thm. Existence uniqueness regularity and boundedness solution to variational problem} is similar to Theorem 5.2 in \citet{GLT:VariationalApproach(2025)}, with an important distinction: the functional \eqref{eq. Functional I} is not convex for $\nu \in (0,1)$, and hence uniqueness is not guaranteed a priori.
We observe that convexity is unnecessary for the existence of a minimiser. Standard estimates and compact embedding arguments yield coercivity and weak sequential lower semicontinuity of \eqref{eq. Functional I}. Inspection of the proof of Theorem 5.2 in \citet{GLT:VariationalApproach(2025)} shows that neither convexity nor uniqueness is required to prove boundedness and $C^2$ regularity of a minimiser. For completeness, we state the following lemma.

\begin{lemma} \label{lem. any minimiser is bounded and C^2}
    Let $g$ be a solution to the minimisation problem \eqref{eq. 1-d minisation problem}. Then $g \in C^2(E;\mathbb{R})$ and $0<g\leq M$ on $E$ for some $M>0$.
\end{lemma}

Thus, Theorem \ref{thm. Existence uniqueness regularity and boundedness solution to variational problem} and Lemma \ref{lem. any minimiser is bounded and C^2} imply there may exist a family $\{g_\alpha\}_{\alpha \in \mathcal{I}}$ of solutions, each obeying $0<g_\alpha\leq M$ and $g_\alpha \in C^2(E;\mathbb{R})$, where $\mathcal{I}$ is some indexing set.
\begin{remark}
  A key insight of this paper is that an \emph{arbitrary} solution to \eqref{eq. 1-d minisation problem} identifies the value function \eqref{eq. agent's goal}. Hence, uniqueness of the value function necessarily implies that the family of solutions $\{g_\alpha\}_{\alpha \in \mathcal{I}}$ to \eqref{eq. 1-d minisation problem} is a singleton. Thus, uniqueness for \eqref{eq. 1-d minisation problem} is proven a posteriori through verification.
\end{remark}

\section{Main result} \label{sec. Main result}
\begin{theorem} \label{thm. Main Theorem}
    Let $g\in C^2(E;\mathbb{R})$ be a solution to the variational problem \eqref{eq. 1-d minisation problem}. The value function is $\hat{V}(x,y)=\frac{x^{1-\gamma}}{1-\gamma}g(y)^{\frac{\gamma}{\gamma+(1-\gamma)\rho^{\top}\rho}}$ and the optimal policies are
    \begin{align*}
    \pi^{\ast}(y)&=\frac{1}{\gamma}\Sigma^{-1}(y)\mu(y)+\frac{1}{\gamma+(1-\gamma)\rho^{\top}\rho}\Sigma^{-1}(y)\Upsilon(y)\frac{ g'(y)}{g(y)}, \nonumber\\
     l^{\ast}(y)&=g(y)^{-\frac{\gamma}{(\gamma+(1-\gamma)\rho^{\top}\rho)\delta\theta}}.
\end{align*}
\end{theorem}

\begin{corollary} \label{cor. solution to the minimisation problem is unique}
    The minimisation problem \eqref{eq. 1-d minisation problem} has a unique solution in $H(E)_+$.
\end{corollary}
\nihil{
\begin{corollary} \label{cor. Impossibility theorem}
    Let $m >0$. Assume there exists a classical solution to the HJB equation \eqref{eq. HJB} of the form $V(x,y)=\frac{x^{1-\gamma}}{1-\gamma}g(y)^m$ satisfying $0<g<M$ on $E$ for some $M>0$. Then $V$ can be the value function if and only if $\theta>0$.
\end{corollary}
}

\subsection{The Hamilton-Jacobi-Bellman equation}
In this section we explain the connection between the variational problem \eqref{eq. 1-d minisation problem} and the stochastic control problem \eqref{eq. agent's goal}.
We conjecture the value function
\begin{align} \label{eq. Value function}
    \hat{V}(x, y) = \sup_{(\pi, l) \in \mathcal{A}} \EE\left[\int_0^{\infty}\frac{(X_t\,l_t)^{1-\delta}}{1-\delta}\left((1-\gamma)V^l_t\right)^{1-\frac{1}{\theta}}dt \, \Big | \, X_0 =  x, Y_0 = y\right],
\end{align}
is of the form $\hat{V}(x,y)=\frac{x^{1-\gamma}}{1-\gamma}g(y)^m$ for some $g:E \to \mathbb{R}_+$, a parameter $m \in \mathbb{R}\setminus\{0\}$ to be chosen, and $(x,y) \in \mathbb{R}_+ \times E$.
The resulting HJB equation in terms of $g$ is then
\begin{align} \label{eq. HJB}
    0= r&+\frac{m\, b}{(1-\gamma)}\frac{g'}{g}+\frac{a^2\,m}{2(1-\gamma)}\frac{g''}{g}+\frac{a^2(m-1)m}{2(1-\gamma)}\left(\frac{g'}{g}\right)^2\\ \nonumber
    &+\sup_{\pi,l}\left\{\frac{l^{1-\delta}}{1-\delta}\,g^{-m/\theta}-l+\pi^{\top}\mu-\frac{\gamma}{2}\pi^\top\Sigma \pi+m\,\pi^{\top}\Upsilon\frac{g'}{g}\right\}.
\end{align}
The first-order conditions yield the candidate optimal controls
\begin{align} \label{eq. Candidate optimal controls}
    \hat{\pi}(y)=\frac{1}{\gamma}\Sigma^{-1}(y)\mu(y)+\frac{m}{\gamma}\Sigma^{-1}(y)\Upsilon(y)\frac{ g'(y)}{g(y)}, && \hat{l}(y)=g(y)^{-\frac{m}{\delta \theta}},
\end{align}
and substitution into \eqref{eq. HJB} yields the semilinear equation
\begin{align}
    \frac{m\,a^2}{2\theta \delta}g''&+\left[\frac{m}{\theta \delta}b - \frac{m}{\gamma}\left(1-\frac{1}{\delta}\right)\left(\Upsilon^\top\Sigma^{-1}\mu\right)\right]g'- \left(1-\frac{1}{\delta}\right)\left(r + \frac{\mu^\top\Sigma^{-1}\mu}{2\gamma}\right)g \\ \nonumber
    &-\left[\frac{a^2(1-m)m}{(1-\delta)\delta \theta}-\frac{m^2}{\delta \gamma}\Upsilon^{\top}\Sigma^{-1}\Upsilon\right]\frac{(1-\delta)}{2}\frac{(g')^2}{g}=-g^{1-\frac{m}{\delta \theta}}
\end{align}
We note that if $\gamma=\delta$ (which forces $\theta=1$) and $m=\gamma$, then one recovers the HJB equation in the CRRA regime \citep{GW:ConsumptionIncompleteMarkets(2020)}. Since the state variable is scalar, $\Upsilon^{\top}\Sigma^{-1}\Upsilon=a^2\rho^{\top}\rho$; thus the prefactor in front of the quadratic non-linearity simplifies, and the equation becomes
\begin{align}
    \frac{m\,a^2}{2\theta \delta}g''&+\left[\frac{m}{\theta \delta}b - \frac{m}{\gamma}\left(1-\frac{1}{\delta}\right)\left(\Upsilon^\top\Sigma^{-1}\mu\right)\right]g'- \left(1-\frac{1}{\delta}\right)\left(r + \frac{\mu^\top\Sigma^{-1}\mu}{2\gamma}\right)g \\ \nonumber
    &-\left[\frac{(1-m)m}{(1-\delta)\delta \theta}-\frac{m^2}{\delta \gamma}\rho^{\top}\rho\right]\frac{a^2(1-\delta)}{2}\frac{(g')^2}{g}=-g^{1-\frac{m}{\delta \theta}}.
\end{align}
Equivalently,
\begin{align}
    \frac{m\,a^2}{2\theta \delta}g''&+\frac{m}{\theta \delta}\left[b - \left(1-\frac{1}{\gamma}\right)\left(\Upsilon^\top\Sigma^{-1}\mu\right)\right]g'- \left(1-\frac{1}{\delta}\right)\left(r + \frac{\mu^\top\Sigma^{-1}\mu}{2\gamma}\right)g \\ \nonumber
    &-\left[\frac{(1-m)m}{(1-\delta)\delta \theta}-\frac{m^2}{\delta \gamma}\rho^{\top}\rho\right]\frac{a^2(1-\delta)}{2}\frac{(g')^2}{g}=-g^{1-\frac{m}{\delta \theta}}.
\end{align}
Using \eqref{eq. tilde b}, this becomes
\begin{align}
    \frac{m\,a^2}{2\theta \delta}g''&+\frac{m}{\theta \delta}\tilde{b}\,g'- \kappa \,g -\left[\frac{(1-m)m}{(1-\delta)\delta \theta}-\frac{m^2}{\delta \gamma}\rho^{\top}\rho\right]\frac{a^2(1-\delta)}{2}\frac{(g')^2}{g}=-g^{1-\frac{m}{\delta \theta}}.
\end{align}

If we set $m=\delta \theta$, then the power non-linearity vanishes and the HJB equation has essentially the same form as in the CRRA regime. In the CRRA case, the HJB equation has proven notoriously difficult to analyse; recent progress has been made by \citet{GHH:ConsumptionINvestmentSFM(2025)}, who provide a general theory of second-order semilinear equations with no initial conditions. Instead, we leverage the fact that the state variable is scalar and employ the transformation from \citet{Thaleia:FS(2001)}. Since $\rho^\top\rho$ is constant, the prefactor in front of the quadratic non-linearity is eliminated by fixing $m$ to be the positive solution of a quadratic equation
\begin{equation} \label{eq. quadratic equation in m}
    \frac{(1-m)m}{(1-\delta)\delta \theta}-\frac{m^2}{\delta \gamma}\rho^{\top}\rho=0.
\end{equation}
As $m\neq0$, we fix 
\begin{equation}
    m=\frac{\gamma}{\gamma+(1-\gamma)\rho^{\top}\rho},
\end{equation}
which solves \eqref{eq. quadratic equation in m}. Recalling the definition of $\nu=\frac{m}{\delta \theta}$ yields a compact representation of the HJB equation:
\begin{equation} \label{eq. general HJB final semi-linear form}
    \frac{\nu a^2}{2}g''+\nu\,\tilde{b}\,g'-\kappa \,g=-g^{1-\nu},
\end{equation}
where $\kappa$ and $\tilde{b}$ are defined in \eqref{eq. Merton consumption kappa} and \eqref{eq. tilde b}.
Now consider the minimisation problem \eqref{eq. 1-d minisation problem}. We observe that the associated Euler-Lagrange equation is exactly the HJB equation \eqref{eq. general HJB final semi-linear form}. Hence, a solution of the variational problem is a classical solution to the HJB equation and provides a candidate value function. Theorem \ref{thm. Main Theorem} then ensures that any solution of the variational problem is the value function, thereby proving that the minimisation problem \eqref{eq. 1-d minisation problem} has a unique solution.

\section{Examples} \label{sec. Examples}
In this section we provide examples of market models with nonlinear, unbounded, and non-affine coefficient structures to which Theorem \ref{thm. Main Theorem} applies. For notational simplicity, we use only a single risky asset, even though all previous results hold for an arbitrary number of risky assets. Numerical solutions are obtained using the scheme outlined in \citet{GLT:VariationalApproach(2025)}.

Throughout \emph{this section only}, we assume $\delta>1$, which guarantees that Assumption \ref{assum. kappa is strictly positive} holds provided the risk-free rate is positive. If $\delta>1$, then $\gamma>1$ is required to ensure $\theta>0$. Finally, imposing $\theta < 1$ yields the parameter regime used in this section:
\begin{align} \label{assum. Example parameters}
    1<\gamma<\delta.
\end{align}
For the numerical solution of each model we reintroduce the discounting parameter $\beta>0$, which was absorbed into the interest rate for the theoretical analysis.

\nihil{
To check the admissible parameter regime for $\delta,\gamma,\rho^2$ we first note $\delta>1$ already forces $\gamma>1$ to ensure $\theta>0$. Thus Assumption \ref{assum. Parameter regime} holds if
\begin{align} 
    \gamma>1; && \delta>1; && \frac{\gamma(\delta\gamma-2\delta+1)}{\delta(\gamma-1)^2}\leq\rho^2.
\end{align}
As noted in the model introduction, Assumption \eqref{assum. Parameter regime}, and hence the above condition, implies $\theta \in (0,1)$. Furthermore, $\theta \in (0,1)$ also forces $\gamma<\delta$. This provides a straightforward, non-restrictive set of sufficient conditions to ensure an admissible parameter regime. We note that in a complete market, $\rho^2=1$, the third condition is equivalent to $\gamma<\delta$. Hence for the remainder of this section we assume \eqref{assum. Example parameters}, which guarantees Assumption \ref{assum. Parameter regime} holds.
}
\subsection{Classical models}
We observe that our framework embeds many classical models, including the CIR state variable used in \citet{Xing:EZConsumption(2017)} and the OU excess-return model \citep{KO:RFS(2015)}.
\subsubsection{CIR State Variable}
The first example features a square-root diffusion state variable driving the interest rate, excess returns, and volatility. We follow the parametrisation adopted by \citet{Xing:EZConsumption(2017)}, where the interest rate fluctuates above a baseline $r_0>0$, and the Sharpe ratio is proportional to the square root of the state variable, with proportionality constant $\lambda$:
\begin{align} \label{eq.Heston model}
      r(Y_t)=r_0+&r_1Y_t; \quad dR_t=\sigma\,\lambda\,Y_tdt+\sigma\,\sqrt{Y_t}dB_t, \nonumber\\
     dY_t&=b_0(b_1-Y_t)dt+a\,\sqrt{Y_t}dW_t.
\end{align}
Here $r_0,r_1,\sigma,a,b_0,b_1,\lambda >0$.
In this model,
\begin{align*}
    \kappa(y)=\left(1-\frac{1}{\delta}\right)\left(r_0+r_1\,y+\frac{\lambda^2 y}{2\gamma }\right);&& \tilde{b}(y)=b_0(b_1-y)-\left(1-\frac{1}{\gamma}\right)\rho a \lambda y,
\end{align*}

\begin{align*}
    \eta(y)=\frac{C}{a^2}\, y^{2b_0 b_1/a^2 - 1}\exp\!\left(-\frac{2[b_0 + (1-\frac{1}{\gamma})\rho a\lambda]}{a^2}\,y\right).
\end{align*}
The Feller condition $2b_0b_1>a^2$ ensures that Assumptions \ref{assum. Market model is well posed}-\ref{assum. Market is arbitrage free} hold. Assumption \ref{assum. kappa is strictly positive} holds because $r_0,r_1>0$ and $\delta>1$ (the standing assumption throughout this section). The condition $b_0 + (1-\frac{1}{\gamma})\rho a\lambda>0$ ensures that Assumptions \ref{assum. eta is a density}-\ref{assum. Recurrence} hold, and thus Theorem \ref{thm. Main Theorem} applies. The optimal policies are
\begin{align*}
    \pi^\ast(y)=\frac{\lambda  }{\gamma\,\sigma} + \frac{a\,\rho}{\sigma(\gamma+(1-\gamma)\rho^2)}\frac{g'(y)}{g(y)}; && l^\ast(y)=g(y)^{\frac{-\gamma}{(\gamma+(1-\gamma)\rho^2)\delta \theta}}.
\end{align*}
The numerical solution for this model is shown in Figure \ref{fig:cir_consumption} with $r_1=0$, i.e., with constant interest rates.

We report the optimal consumption ratio for various combinations of $\gamma$ and $\delta$ in Figure \ref{fig:cir_consumption}. In each case, consumption is essentially constant across states, in contrast to the linear myopic policy, which is increasing in the state. For fixed $\gamma$, the consumption ratio decreases as $\delta$ increases.
In particular, the consumption ratio is more sensitive to changes in $\delta$ with fixed $\gamma$ than to changes in $\gamma$ with fixed $\delta$. For $\gamma=2$, the consumption ratio decreases by a factor of approximately $0.93$ as $\delta$ increases from $3$ to $5$. In contrast, with $\delta=5$, the consumption ratio decreases by a factor of approximately $0.99$ as $\gamma$ increases from $2$ to $4$. Hence, changes in risk aversion have a proportionally smaller effect on the optimal consumption ratio than changes in $\delta$. These findings are similar to those in \citet{Xing:EZConsumption(2017)} for the finite-horizon setting.

\begin{figure}[H]
  \begin{minipage}[t]{0.48\textwidth}
    \includegraphics[width=\textwidth, height=5.5cm]{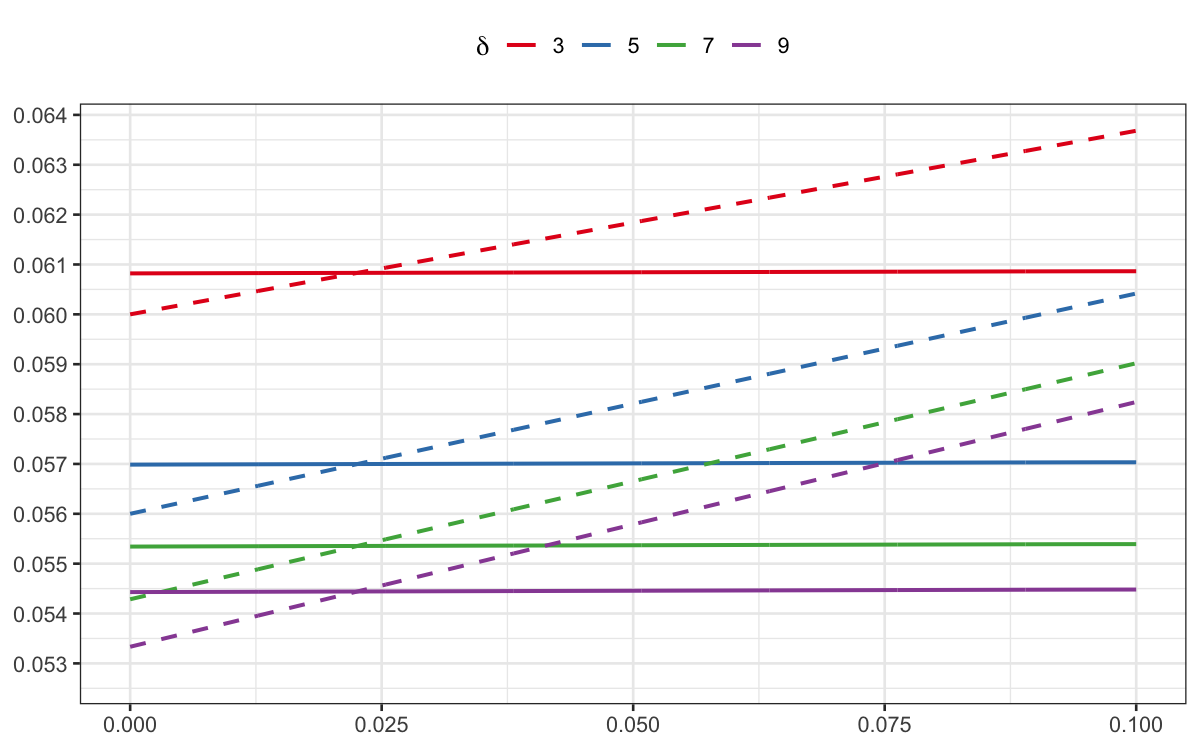}
  \end{minipage}
  \hfill
  \begin{minipage}[t]{0.48\textwidth}
    \includegraphics[width=\textwidth,height=5.5cm]{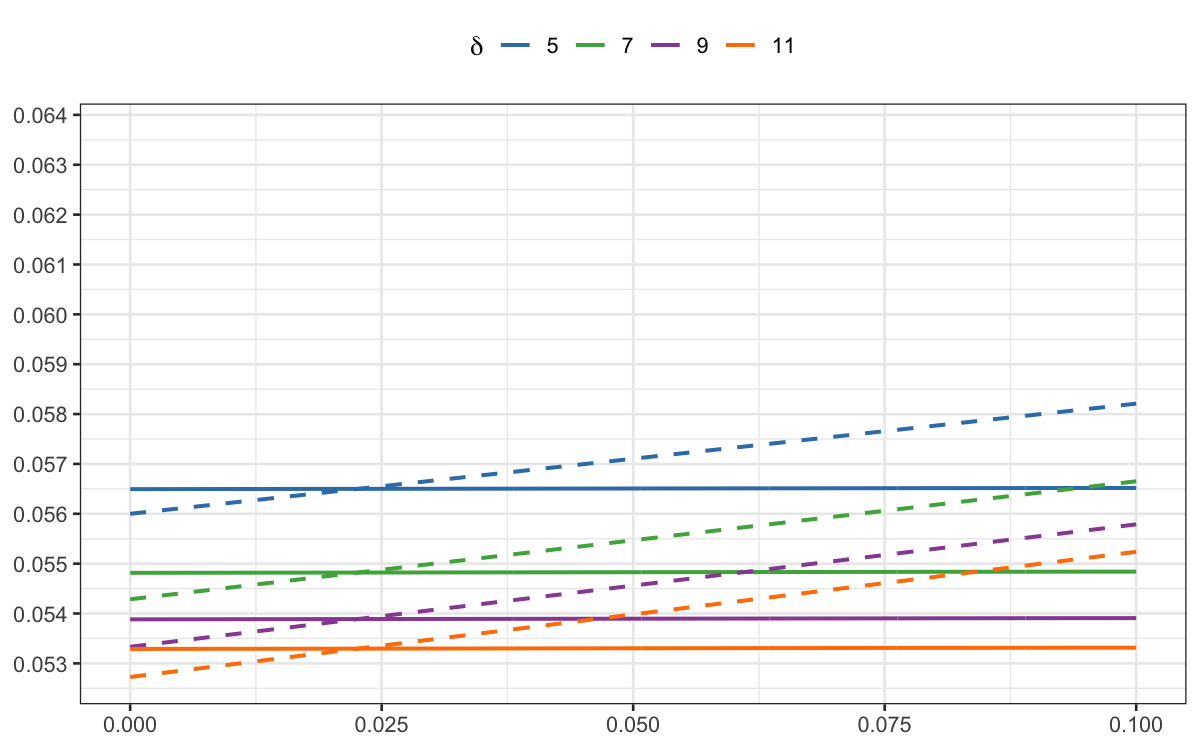}
  \end{minipage}

  \caption{Optimal consumption ratio $l^\ast(y) = g(y)^{-\nu}$ (solid lines) and the myopic
    benchmark $\kappa(y)$ (dashed lines) as a function of the
    CIR state variable $y$ for $\gamma=2$ (left panel) and $\gamma=4$ (right panel).
    Parameters: $r_0=0.05$, $r_1=0$, $\sigma=1$, $\beta=0.08$, $\lambda=0.47$, $b_0=5$, $b_1=0.0225$,
    $a=0.25$, $\rho=-0.5$ taken from \citet{Xing:EZConsumption(2017)}.}
  \label{fig:cir_consumption}
\end{figure}

\subsubsection{Kim-Omberg excess returns}
Next, we present the classical Gaussian excess-return model from \citet{KO:RFS(2015)}. Let $r>0$ and
\begin{align} \label{eq. Kim-Omberg}
      dR_t&=Y_tdt+\sigma\,dB_t, \nonumber\\
     dY_t&=b(\bar{\mu}-Y_t)dt+a dW_t.
\end{align}
In this model, $\bar{\mu},\sigma,b,a >0$,
\begin{align*}
    \kappa(y)=\left(1-\frac{1}{\delta}\right)\left(r+\frac{y^2}{2\gamma \sigma^2}\right);&& \tilde{b}(y)=b\bar{\mu}-y\left(b+\left(1-\frac{1}{\gamma}\right)\left(\frac{\rho a }{\sigma}\right)\right),
\end{align*}

\begin{align*}
    \eta(y)=\frac{C}{a^2}\exp\!\left(-\frac{\lambda y^2}{a^2}+\frac{2b\bar{\mu}\,y}{a^2}\right),
\end{align*}
where $\lambda \coloneqq\left(b+\left(1-\frac{1}{\gamma}\right)\left(\frac{\rho a }{\sigma}\right)\right)>0$. Assumptions \ref{assum. Market model is well posed}-\ref{assum. Market is arbitrage free} hold, while $\delta>1$ ensures that Assumption \ref{assum. kappa is strictly positive} holds. Finally, $\lambda>0$ is sufficient to ensure that Assumptions \ref{assum. eta is a density}-\ref{assum. Recurrence} hold. Hence Theorem \ref{thm. Main Theorem} gives the optimal policies:

\begin{align*}
    \pi^\ast(y)=\frac{y }{\gamma\,\sigma^2} + \frac{a\,\rho}{\sigma(\gamma+(1-\gamma)\rho^2)}\frac{g'(y)}{g(y)}; && l^\ast(y)=g(y)^{\frac{-\gamma}{(\gamma+(1-\gamma)\rho^2)\delta \theta}}.
\end{align*}

\begin{figure}[H]
  \begin{minipage}[t]{0.48\textwidth}
    \includegraphics[width=\textwidth, height=5.5cm]{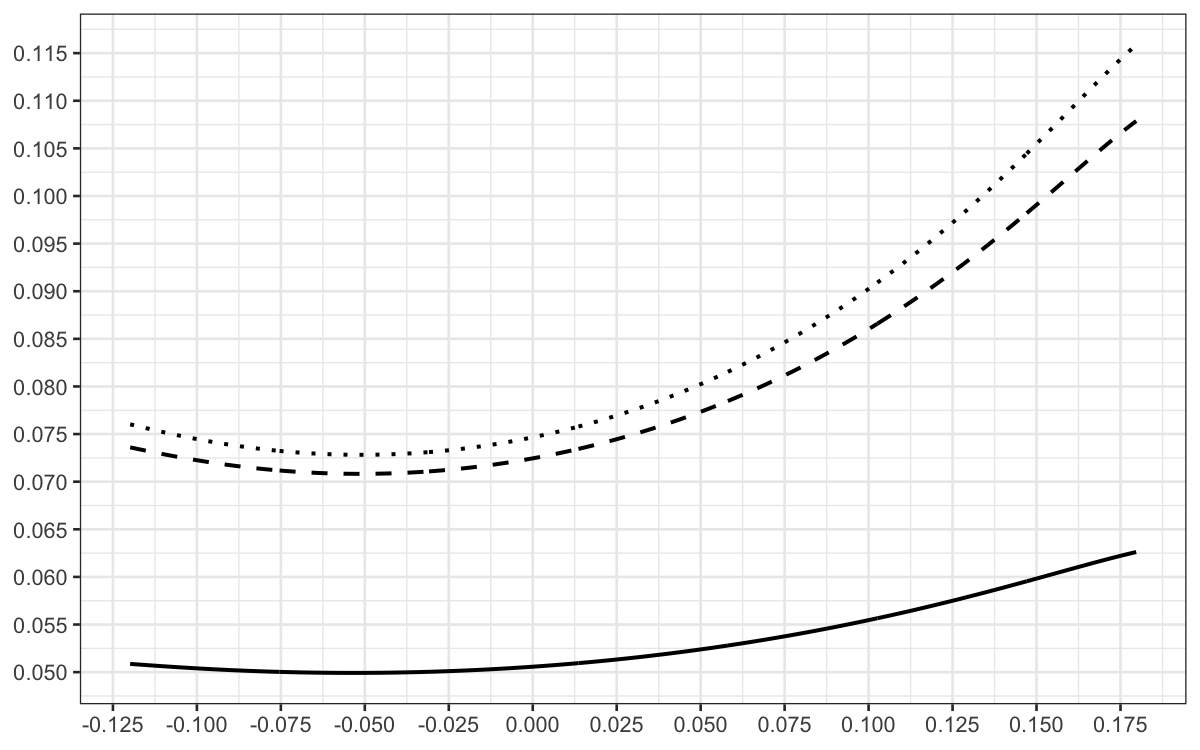}
  \end{minipage}
  \hfill
  \begin{minipage}[t]{0.48\textwidth}
    \includegraphics[width=\textwidth,height=5.5cm]{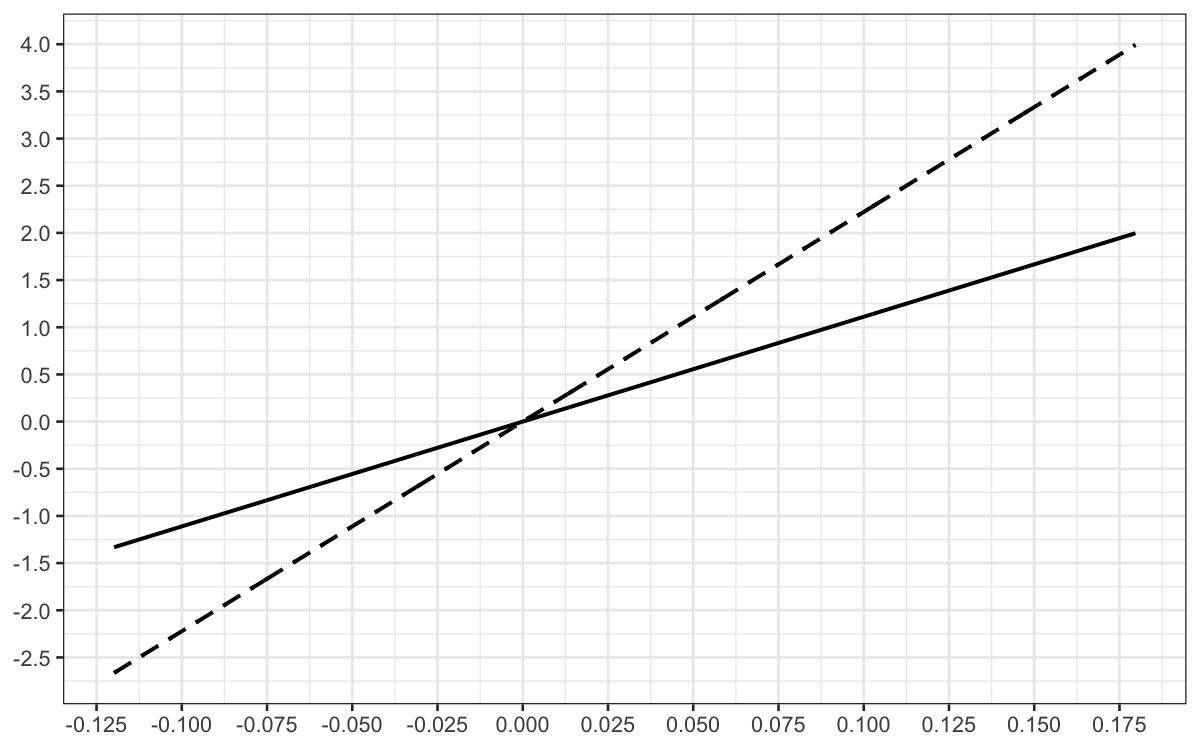}
  \end{minipage}

  \caption{Optimal consumption ratio $l^\ast(y) = g(y)^{-\nu}$ (left panel) and optimal portfolio weight $\pi^\ast(y)$ (right panel) as a function of the state variable $y$ (within a $99\%$ confidence interval for the long-run stationary density of $Y$), for $(\gamma,\delta)=(2,9)$ (dotted lines), $(\gamma,\delta)=(2,5)$ (dashed lines) and $(\gamma,\delta)=(4,5)$ (solid lines). Parameters: $r=0.02$, $\beta=0.05$, $\sigma=0.15$, $b=0.50$, $\bar{\mu}=0.06$,
    $a=0.05$, $\rho=-0.30$.}
  \label{fig:ko_consumption_portfolio}
\end{figure}
In Figure \ref{fig:ko_consumption_portfolio}, we present both the optimal consumption ratio and the optimal portfolio weight. We observe that $\delta$ has no effect on the investment policy; with $\gamma=2$, the investment policies for $\delta=5$ and $\delta=9$ are identical. In contrast to the CIR model, changes in $\gamma$ have a larger relative effect on the optimal consumption ratio than changes in $\delta$. Consumption is more sensitive to the state variable than in the CIR model, but for higher levels of risk aversion ($\gamma=4$), consumption is very stable, ranging between $5\%$ and $6.25\%$ across all states.

\subsection{Fat-Tailed Excess Returns}
We next consider a model with $r>0$:
\begin{align} \label{eq. FT excess returns}
      dR_t&=Y_tdt+\sigma\,dB_t, \nonumber\\
     dY_t&=b(\bar{\mu}-Y_t)^3dt+a\left(1+\frac{(Y_t-\bar{\mu})^2}{\alpha^2}\right)dW_t.
\end{align}
In this model, $\bar{\mu},\sigma,b,a >0$,
\begin{align*}
    \kappa(y)=\left(1-\frac{1}{\delta}\right)\left(r+\frac{y^2}{2\gamma \sigma^2}\right);&& \tilde{b}(y)=b(\bar{\mu}-y)^3-\left(1-\frac{1}{\gamma}\right)\frac{\rho a}{\sigma}\left(1+\frac{(y-\bar{\mu})^2}{\alpha^2}\right)y,
\end{align*}
\begin{align*}
        \eta(y)&=\frac{C}{a^2\left(1+\frac{(y-\bar{\mu})^2}{\alpha^2}\right)^{2+\lambda+\left(1-\frac{1}{\gamma}\right)\frac{\rho \alpha^2}{a\sigma}}}\exp\!\left(\frac{-\lambda}{\left(1+\frac{(y-\bar{\mu})^2}{\alpha^2}\right)}-2\left(1-\frac{1}{\gamma}\right)\frac{\rho\bar{\mu} \alpha}{a\sigma}\arctan\left(\frac{y-\bar{\mu}}{\alpha}\right) \right),\\
\end{align*}
where $\lambda=\frac{b\alpha^4}{a^2}$  with $\lambda+\left(1-\frac{1}{\gamma}\right)\frac{\rho \alpha^2}{a\sigma}>-\frac{3}{2}$, and $C>0$ is a normalising constant.
It is straightforward to check that Assumptions \ref{assum. Market model is well posed}-\ref{assum. Recurrence} hold. The optimal policies are
\begin{align*}
    \pi^\ast(y)=\frac{y }{\gamma\,\sigma^2} + \frac{a\,\rho\,\left(1+\frac{(y-\bar{\mu})^2}{\alpha^2}\right)}{\sigma(\gamma+(1-\gamma)\rho^2)}\frac{g'(y)}{g(y)}; && l^\ast(y)=g(y)^{\frac{-\gamma}{(\gamma+(1-\gamma)\rho^2)\delta \theta}}.
\end{align*}

 To ensure a fair comparison between models \eqref{eq. Kim-Omberg} and \eqref{eq. FT excess returns}, we calibrate the parameters for model \eqref{eq. FT excess returns} by matching the first two moments of the stationary density of the state variable in model \eqref{eq. Kim-Omberg} and by matching the derivative at zero of the autocorrelation function under the stationary distribution. The state variable and stationary density in model \eqref{eq. Kim-Omberg} are
\begin{align*}
    dY_t=0.5(0.06-Y_t)dt+0.05\,dW_t; && \Phi_{OU} \sim \mathcal{N}(0.06,0.05^2).
\end{align*}
Consider the state variable from model \eqref{eq. FT excess returns} and its stationary density:
\begin{align*}
    dY_t=b(\bar{\mu}-Y_t)^3dt+a\left(1+\frac{(Y_t-\bar{\mu})^2}{\alpha^2}\right)dW_t; && \Phi(y) \propto \frac{1}{\left(1+\frac{(y-\bar{\mu})^2}{\alpha^2}\right)^{2+\lambda}}\exp\!\left(\frac{-\lambda}{\left(1+\frac{(y-\bar{\mu})^2}{\alpha^2}\right)} \right).
\end{align*}
We need to fix the parameters $(\bar{\mu},\alpha,a,b)$. We set $\bar{\mu}=0.06$ to match the first moment and fix $\lambda=1$, which determines the tail index.
Using the transformation $z=(y-\bar{\mu})/\alpha$, we see that the stationary density satisfies
\begin{equation*}
    z \mapsto \frac{C}{\left(1+z^2\right)^{3}}\exp\left(\frac{-1}{\left(1+z^2\right)} \right)
\end{equation*}
where $C$ is a normalising constant independent of $\alpha$ ($1/C \approx 0.5218086$). Undoing this transformation then yields the exact form of the stationary density:
\begin{equation*}
        \Phi(y) = \frac{C}{\alpha\left(1+\frac{(y-0.06)^2}{\alpha^2}\right)^{3}}\exp\!\left(\frac{-1}{\left(1+\frac{(y-0.06)^2}{\alpha^2}\right)} \right).
\end{equation*}
Next we match the variance:
\begin{align*}
    0.05^2&=\int_{-\infty}^{\infty}\frac{C(y-0.06)^2}{\alpha\left(1+\frac{(y-0.06)^2}{\alpha^2}\right)^{3}}\exp\!\left(\frac{-1}{\left(1+\frac{(y-0.06)^2}{\alpha^2}\right)} \right)dy\\
    & =\alpha ^2 \int_{-\infty}^{\infty} \frac{C\,z^2}{\left(1+z^2\right)^{3}}\exp\left(\frac{-1}{\left(1+z^2\right)} \right)dz\\
    & = \alpha^2 m_2,
\end{align*}
where $m_2$ can be estimated numerically ($m_2 \approx 0.4708723$) to yield $\alpha = \frac{0.05}{\sqrt{m_2}} \approx 0.07286491$. Next, to calibrate the parameter $b$ we match the derivative at zero of the autocorrelation function under the stationary distribution. Recall for model \eqref{eq. Kim-Omberg} this yields $R'(0)=-0.5$ where
\begin{align*}
    R(t)=\frac{\mathbb{E}[Y_0Y_t]}{Var(Y_0)}.
\end{align*}
For model \eqref{eq. FT excess returns} we see
\begin{equation*}
    R'(0)=\frac{\mathbb{E}[Y_0b(\bar{\mu}-Y_0)^3]}{Var(Y_0)}.
\end{equation*}
Hence we set
\begin{equation*}
    b= -\frac{Var(Y_0)}{2\mathbb{E}[Y_0(\bar{\mu}-Y_0)^3]}.
\end{equation*}
The right-hand side can be estimated numerically, which yields $b \approx 22.82987$. Finally, we use $\lambda=\frac{b\alpha^4}{a^2}$ to solve for $a$ (recall $\lambda=1$), yielding the final parameters:

\begin{align*}
 \bar{\mu}=0.06, && \alpha =  0.072865, && a= 0.025368, && b = 22.82987.  
\end{align*}
Figure \ref{fig:ft_consumption_portfolio} shows the numerical solution for model \eqref{eq. FT excess returns}.
\begin{figure}[H]
  \begin{minipage}[t]{0.48\textwidth}
    \includegraphics[width=\textwidth, height=5.5cm]{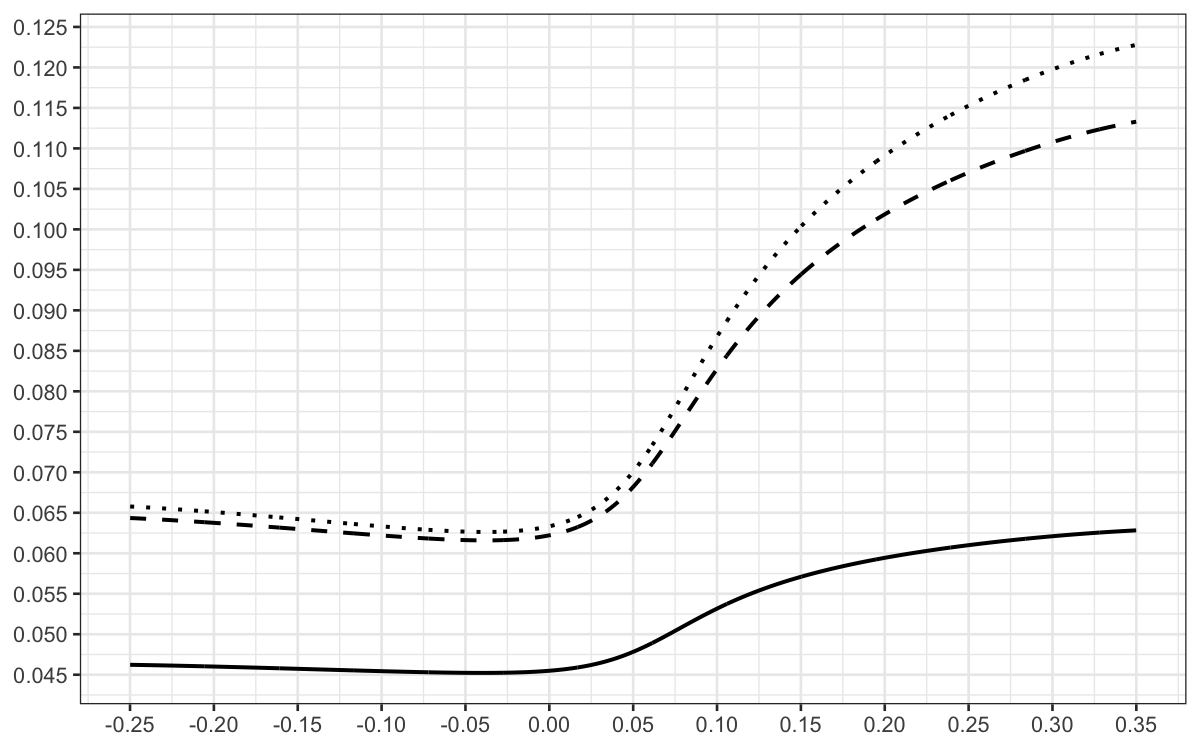}
  \end{minipage}
  \hfill
  \begin{minipage}[t]{0.48\textwidth}
    \includegraphics[width=\textwidth,height=5.5cm]{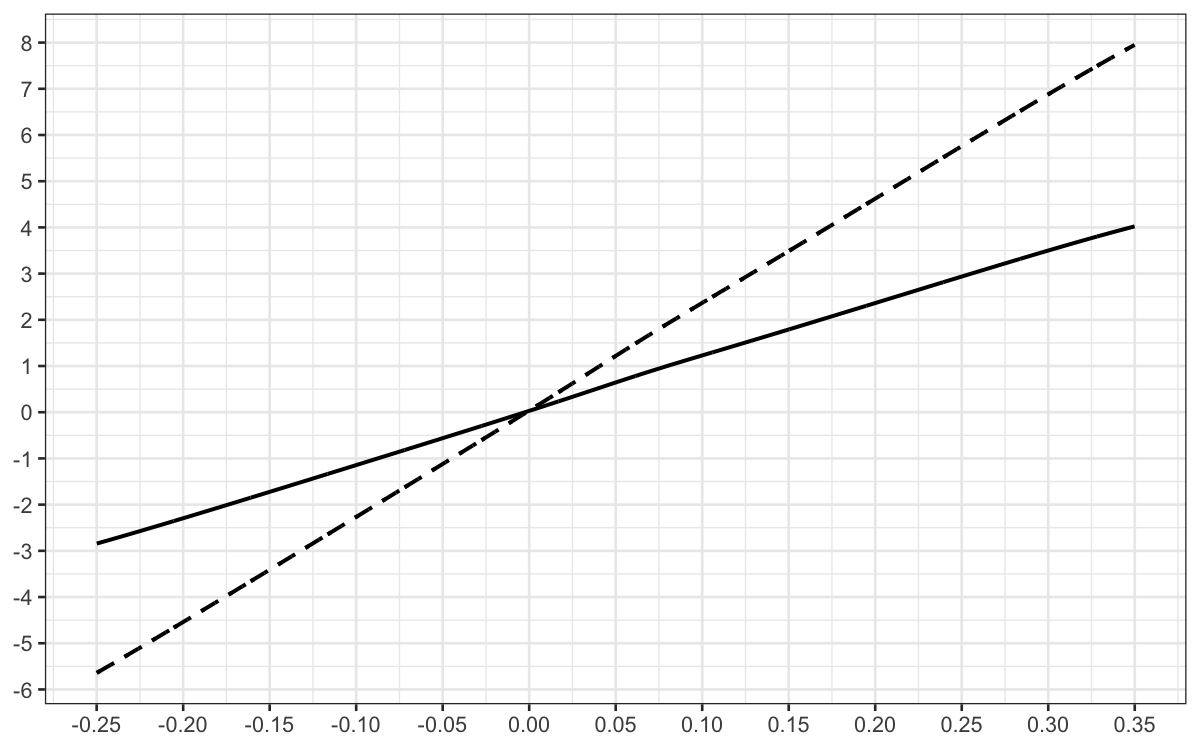}
  \end{minipage}

  \caption{Optimal consumption ratio $l^\ast(y) = g(y)^{-\nu}$ (left panel) and optimal portfolio weight $\pi^\ast(y)$ (right panel) as a function of the state variable $y$ (within a $99\%$ confidence interval for the long-run stationary density of $Y$), for $(\gamma,\delta)=(2,9)$ (dotted lines), $(\gamma,\delta)=(2,5)$ (dashed lines) and $(\gamma,\delta)=(4,5)$ (solid lines). Parameters: $r=0.02$, $\beta=0.05$, $\sigma=0.15$, $b=22.82987$, $\bar{\mu}=0.06$,
    $a=0.025368$, $\alpha=0.072865$, $\rho=-0.30$.}
  \label{fig:ft_consumption_portfolio}
\end{figure}

In Figure \ref{fig:ft_consumption_portfolio}, to facilitate comparison with markets exhibiting fat-tailed excess returns, we report the optimal consumption and investment policies with the same preference parameters as in the Gaussian excess-returns model \eqref{eq. Kim-Omberg}. The relative shape of the investment policy is similar to that of the Kim-Omberg model; however, the consumption policy is significantly different.

The consumption plot illustrates an abrupt departure from the quadratic myopic policy. Consumption is extremely stable when the expected excess return is negative, varying by less than half a percent in all parameter configurations. When $y>0$, the consumption policy is increasing in the state. In all three parameter regimes there is a critical point at which the consumption policy switches from convex to concave in the state. This phenomenon does not occur in the Kim-Omberg model \eqref{eq. Kim-Omberg}, where the consumption policy is a convex function of the state over the entire $99\%$ confidence interval for the stationary distribution.

\section{Verification} \label{sec. Verification}

The goal of this section is to prove Theorem \ref{thm. Main Theorem}. Throughout this section, let $g$ be a minimiser of the variational problem \eqref{eq. 1-d minisation problem} (which exists by Theorem \ref{thm. Existence uniqueness regularity and boundedness solution to variational problem}). By Lemma \ref{lem. any minimiser is bounded and C^2},  $g \in C^2(E;\mathbb{R})$ and $0<g\leq M$; since the Euler-Lagrange equation coincides with the HJB equation, the candidate value function is $\hat{V}(x,y)=\frac{x^{1-\gamma}}{1-\gamma}g(y)^m$. Verification has three steps. First, we prove that the candidate value function is the unique utility process associated with the candidate optimal consumption stream. Next, we introduce the perturbed process $\hat{V}(X+\varepsilon\hat{X},Y)$, where $X$ is the wealth process generated by an arbitrary consumption stream and $\hat{X}$ is the wealth process under the candidate optimal controls. We prove that $\hat{V}(X+\varepsilon\hat{X},Y)$ is a supersolution to \eqref{eq. Utility Process} for an arbitrary pair of admissible policies; finally, a comparison argument yields optimality.

\begin{prop} \label{prop. Candidate Value function solves BSDE}
    The candidate value function $\hat{V}(x,y)=\frac{x^{1-\gamma}}{1-\gamma}g(y)^m$ with candidate controls $(\hat{\pi},\hat{l})$ (given by \eqref{eq. Candidate optimal controls}) satisfies
    \begin{equation*}
    \hat{V}(x,y)= \EE\left[\int_0^{\infty}\frac{(X_t^{\hat{\pi},\hat{l}}\,\hat{l}_t)^{1-\delta}}{1-\delta}\left((1-\gamma)\hat{V}(X_t^{\hat{\pi},\hat{l}},Y_t)\right)^{1-\frac{1}{\theta}}dt \, \Big | \, X_0 =  x, Y_0 = y\right],
    \end{equation*}
    and hence is the unique solution to the BSDE \eqref{eq. Utility Process} for the candidate controls.
\end{prop}

Introduce the process $Z_t^{\varepsilon}=X_t^{\pi,l}+\varepsilon \hat{X}_t$, where $\varepsilon>0$, $(\pi,l) \in \mathcal{A}$ are arbitrary, and $\hat{X}$ is the wealth process under the candidate optimal controls $(\hat{\pi},\hat{l})$ started from unit wealth (i.e., $\hat{X}_0=1$). Then the dynamics of $Z^{\varepsilon}$ are
\begin{align} \label{eq. Dynamics of Z_t^varepsilon}
    \frac{dZ_t^{\varepsilon}}{Z_t^{\varepsilon}}=r+(\pi^{z})^{\top}dR_t-l^{z}dt,
\end{align}
where
\begin{align*}
    (\pi_t^z)^{\top}=\frac{X_t\pi_t^{\top}+\varepsilon \hat{X}_t\hat{\pi}_t^{\top}}{Z_t^{\varepsilon}};&& l_t^z=\frac{X_tl_t+\varepsilon \hat{X}_t\hat{l}_t}{Z_t^{\varepsilon}}=\frac{C_t^\varepsilon}{Z_t^\varepsilon}.
\end{align*}
We prove that $\hat{V}(Z_t^\varepsilon,Y_t)$ is a supersolution of \eqref{eq. Utility Process} for the consumption stream $C_t^\varepsilon$ in the sense of \citet[Definition 5.3]{HHJ(2023):EZutilityII}. For the reader's convenience, we recall the definition.
\begin{definition}[\cite{HHJ(2023):EZutilityII}]
    Let $C_t$ be a non-negative progressively measurable consumption stream.
    \begin{itemize}
        \item $V$ is a supersolution to \eqref{eq. Utility Process} for the consumption stream $C$ if $\liminf_{t\to \infty}\mathbb{E}[V_{t_+}]\geq0$ and for all bounded stopping times $\tau_1\leq \tau_2$,
    \begin{align*}
        V_{\tau_1}\geq \mathbb{E}\left[V_{\tau_2}+ \int_{\tau_1}^{\tau_2}\frac{(C_t)^{1-\delta}}{1-\delta}\left((1-\gamma)V_t\right)^{1-\frac{1}{\theta}}dt \mid\mathcal{F}_{\tau_1}\right].
    \end{align*}

    \item $V$ is a subsolution to \eqref{eq. Utility Process} for the consumption stream $C$ if $\limsup_{t\to \infty}\mathbb{E}[V_{t_+}]\leq0$ and for all bounded stopping times $\tau_1\leq \tau_2$,
    \begin{align*}
        V_{\tau_1}\leq \mathbb{E}\left[V_{\tau_2}+ \int_{\tau_1}^{\tau_2}\frac{(C_t)^{1-\delta}}{1-\delta}\left((1-\gamma)V_t\right)^{1-\frac{1}{\theta}}dt \mid\mathcal{F}_{\tau_1}\right].
    \end{align*}
    \end{itemize}

\end{definition}
\begin{lemma} \label{lem. Perturbed candidate value function is supersolution}
    The process $\hat{V}(Z_t^\varepsilon,Y_t)$ is a supersolution of equation \eqref{eq. Utility Process} for the consumption stream $C_t^\varepsilon$.
\end{lemma}

\begin{corollary} \label{cor. Perturbed candidate value function is supersolution for arbitrary consumption policies}
    The process $\hat{V}(Z_t^\varepsilon,Y_t)$ is a supersolution of equation \eqref{eq. Utility Process} for any arbitrary consumption stream $C_t$.
\end{corollary}

The proof of Corollary \ref{cor. Perturbed candidate value function is supersolution for arbitrary consumption policies} follows directly from Lemma \ref{lem. Perturbed candidate value function is supersolution} and the observation that the aggregator $f(c,v)$ is non-increasing in its first argument.
The proof of Lemma \ref{lem. Perturbed candidate value function is supersolution} follows from an extension of the arguments in \citet[Theorem 8.1]{HHJ(2023):EZutilityII}.

\begin{lemma} \label{lem. Comparison lemma}
    Let $C$ be an arbitrary consumption stream and assume $\gamma>1$. Let $V$ be a subsolution of \eqref{eq. Utility Process} for the consumption stream $C$. Then $V_{\tau}\leq \hat{V}(Z_{\tau}^\varepsilon,Y_\tau)$ $\mathbb{P}$-a.s. for all finite stopping times $\tau$.
\end{lemma}

\begin{remark}

      Lemma \ref{lem. Comparison lemma} is not a general comparison theorem for arbitrary sub- and supersolutions. It is tailored to the candidate supersolution above, for which a lower bound by a uniformly integrable martingale replaces the assumption that either the subsolution or the supersolution is an element of $\mathbb{UI}(f,c)$ used in \citet[Theorem 5.8]{HHJ(2023):EZutilityII}. This is useful in stochastic-investment-opportunity models, where the sufficient integrability conditions used in the Black–Scholes–Merton case are not easy to verify.
      Recall that a progressively measurable process satisfies $V \in \mathbb{UI}(f,c)$ if
    \begin{equation*}
        \mathbb{E}\Big[\int_0^\infty|f(c_s,V_s)|ds\Big]<\infty,
    \end{equation*}
    and $V$ is uniformly integrable. 

\end{remark}

We are now in a position to prove our main result.
\begin{proof}[Proof of Theorem \ref{thm. Main Theorem}]
    By Corollary \ref{cor. Perturbed candidate value function is supersolution for arbitrary consumption policies}, $\hat{V}(X_t^{\pi,l}+\varepsilon \hat{X}_t,Y_t)$ is a supersolution to equation \eqref{eq. Utility Process} for an arbitrary consumption stream $C_t$. Let $V_t ^{C}$ denote the unique generalised utility process associated with the consumption stream $C_t$.

    First, let $\gamma\in (0,1)$, which implies $V^C \in [0,\infty]$. In this case $V_t^C$ is the minimal supersolution \citep[Theorem 6.5]{HHJ(2023):EZutilityII}. Evaluating these utility processes at time zero yields
    \begin{equation*}
        \hat{V}(x+\varepsilon,y)\geq V_0^C=V_0^{\pi,l}
    \end{equation*}
    by minimality.

    Next, let $\gamma>1$, which implies $V^C \in [-\infty,0]$. In this case $V_t^{C}$ is the maximal subsolution for the consumption stream $C_t$ \citep[Theorem 6.5]{HHJ(2023):EZutilityII} (hence, the argument in the $\gamma\in (0,1)$ case is not applicable). To compare these two utility processes, we use Lemma \ref{lem. Comparison lemma}, yielding
    \begin{equation*}
    \hat{V}(x+\varepsilon,y)\geq V_0^C=V_0^{\pi,l}.
    \end{equation*}
    Thus, in both cases, taking the supremum over $(\pi,l)\in \mathcal{A}$ on both sides and sending $\varepsilon \to0$ yields
    \begin{align*}
        \hat{V}(x,y)\geq \sup_{(\pi,l)\in \mathcal{A}}V_0^{\pi,l}.
    \end{align*}
\end{proof}

\section{Conclusion} \label{sec. Conclusion}
In this article we have solved the infinite-horizon Epstein-Zin optimal consumption-investment problem in an incomplete market with stochastic investment opportunities in the regime $\theta\in(0,1)$. Our main result characterises the value function as the minimiser of a non-convex functional whose Euler-Lagrange equation coincides with the HJB equation; this analysis extends results in \citet{GLT:VariationalApproach(2025)} for the CRRA regime. Verification identifies an arbitrary minimiser with the value function. The argument combines a change of measure to the myopic probability, the uniqueness theory for Epstein-Zin BSDEs of \citet{HHJ(2023):EZutilityII}, and a perturbation argument adapted to the setting of stochastic investment opportunities.

The variational formulation requires only H\"older regularity and mild joint-integrability conditions on the model coefficients, together with positive myopic consumption, and therefore accommodates highly non-linear and non-affine models. The fat-tailed excess-return model of Section \ref{sec. Examples} illustrates this scope; the optimal consumption ratio departs sharply from the myopic policy and exhibits a convex-to-concave transition in the state that is absent under Gaussian returns. This comparison suggests that distributional assumptions on excess returns, routinely imposed for tractability, materially affect optimal consumption policies.

\appendix
\section{Proofs}
\begin{lemma} \label{lem. weak sequential continuity of f}
Let $\nu \in (0,1)$. Assume that $g_n$ is a non-negative weakly convergent sequence in $H(E)$. Then $g_n^{2-\nu} \to g^{2-\nu}\,\,\text{ in }L^1_{\eta}(E)$.
\end{lemma}
\begin{proof}[Proof of Lemma \ref{lem. weak sequential continuity of f}]
\medskip
\noindent\emph{Step 1. (Uniform integrability and tightness.)}
Set $p\coloneqq \frac{2}{2-\nu}$ and $q \coloneqq 2/\nu$; H\"older's inequality yields
\begin{align*}
    \int_A g^{2-\nu} \eta \leq \|g\|_{{L^2}_{\eta}}^{2/p}\,\left(\eta(A)\right)^{\nu/2},
\end{align*}
for any measurable subset $A \subseteq E$. As $g \in H(E)$, this necessarily implies $g \in L^2_{\eta}(E)$, and so the right-hand side is finite. Hence, $g^{2-\nu} \in L^1_{\eta}(E)$. The same estimate and boundedness of $\|g_n\|_{H(E)}$ imply that $g_n^{2-\nu}$ is both uniformly integrable and tight.

\medskip
\noindent\emph{Step 2. (Pointwise convergence)}
Let $(E_m)_{m\geq1}$ be a compact exhaustion of $E$. For each $m\geq 1$, $H(E_m) \hookrightarrow \hookrightarrow L^{2-\nu}_\eta(E_m)$.
Hence, weak convergence in $H(E_m)$ together with the compact embedding yields $g_n \to g$ in $L^{2-\nu}_\eta(E_m)$. Along a subsequence, $g_{n_k} \to g$ pointwise a.e. on $E_m$.
Continuing in this manner, we may construct a diagonal subsequence such that $g_{n_{k_j}} \to g$ a.e. on $E$.

\medskip
\noindent\emph{Step 3. (Vitali convergence theorem)} $g_{n_{k_j}}^{2-\nu}$ is uniformly integrable, tight, and $g_{n_{k_j}}^{2-\nu} \to g^{2-\nu}$ a.e. on $E$. Hence, by the Vitali Convergence Theorem, $g_{n_{k_j}}^{2-\nu} \to g^{2-\nu}$ in $L^1_\eta(E)$. Finally, the subsequence principle in metric spaces yields $g_n^{2-\nu} \to g^{2-\nu}$ in $L^1_\eta(E)$ for the original sequence.
    
\end{proof}

\begin{proof}[Proof of Theorem \ref{thm. Existence uniqueness regularity and boundedness solution to variational problem}]

Under the correspondence $D \leftrightarrow 1/\nu$ and $\alpha \leftrightarrow 2-\nu$, the functional in \citet{GLT:VariationalApproach(2025)} has the same form as \eqref{eq. Functional I}. Hence, in the convex regime $\nu\geq1$, \citet[Theorem 5.2]{GLT:VariationalApproach(2025)} yields existence, boundedness, and regularity of a minimiser.
It remains to treat $\nu \in (0,1)$, where the functional \eqref{eq. Functional I} is non-convex.

\medskip
\noindent\emph{Step 1 (Boundedness from below and coercivity).}
As in \citet{GLT:VariationalApproach(2025)}, introduce
\begin{equation}
    q_y(x)\coloneqq\kappa(y)x^2-\frac{2x^{2-\nu}}{2-\nu}.
\end{equation}
This function yields a unique minimiser at $x=\kappa(y)^{-1/\nu}$. Substituting yields
\begin{equation*}
    q_y(\kappa(y)^{-1/\nu})=-\frac{\kappa(y)^{-\left(\frac{2}{\nu}-1\right)}}{\left(\frac{2}{\nu}-1\right)}.
\end{equation*}
Define $\inf_{y \in E}\kappa(y)=\bar{\kappa}>0$. Thus,
\begin{align*}
    I(g)=\int_E\left(\frac{\nu}{2}a(y)^2(g'(y))^2+q_y(g(y))\right)\eta(y)\geq -\int_E \frac{\kappa(y)^{-\left(\frac{2}{\nu}-1\right)}}{\left(\frac{2}{\nu}-1\right)}\eta(y)\geq -\frac{\bar{\kappa}^{-\left(\frac{2}{\nu}-1\right)}}{\left(\frac{2}{\nu}-1\right)}>-\infty.
\end{align*}
Thus, $\inf_{g \in H(E)_+}I(g)>-\infty$. For coercivity, we define $p\coloneqq\frac{2}{2-\nu}$ and $q\coloneqq 2/\nu$, noting that $1/p+1/q=1$. Then, by Young's inequality,
\begin{align*}
    \frac{2}{2-\nu}x^{2-\nu}=\frac{2}{2-\nu}\left(\frac{\bar{\kappa}}{2}\right)^{-1/p}x^{2-\nu}\left(\frac{\bar{\kappa}}{2}\right)^{1/p} \leq\frac{\bar{\kappa}}{2}x^2 + \left(\frac{2}{2-\nu}\left(\frac{\bar{\kappa}}{2}\right)^{-1/p}\right)^q.
\end{align*}
Define $C_{\bar{\kappa},\nu}\coloneqq \left(\frac{2}{2-\nu}\left(\frac{\bar{\kappa}}{2}\right)^{-1/p}\right)^q>0$. Hence,
\begin{align*}
    I(g)& \geq \int_E \left(\frac{\nu}{2}a(y)^2(g'(y))^2+\kappa(y)g(y)^2-\frac{\bar{\kappa}}{2}g(y)^2-C_{\bar{\kappa},\nu}\right)\eta(y)\\
    & \geq \int_E \left(\frac{\nu}{2}a(y)^2(g'(y))^2+\frac{\bar{\kappa}}{2}g(y)^2\right)\eta(y)-C_{\bar{\kappa},\nu}\\
    & \geq \min\left(\frac{\nu}{2}, \frac{\bar{\kappa}}{2}\right)\|g\|^2_{H(E)}-C_{\bar{\kappa},\nu}.
\end{align*}
Coercivity of $I$ ensures that any minimising sequence is bounded in $H(E)$.

\medskip
\noindent\emph{Step 2 (Existence).} Let $m\coloneqq \inf_{g \in H(E)_+}I(g)$, and let $g_n$ be a minimising sequence in $H(E)_+$. Coercivity of $I$ ensures $\|g_n\|_{H(E)}\leq M$. Assumption \ref{assum. Market model is well posed} ensures $\frac{1}{a^2\eta},\frac{1}{\eta}\in L^1_{\rm loc}(E)$, so $H(E)$ is a Banach space \cite[Theorem 1.11]{KO:WeightedSobSpace(1984)}; since $H(E)$ is a (weighted) $W^{1,2}$-space, it is reflexive. Hence, passing to a subsequence (not relabelled),
\[
g_n\rightharpoonup g \quad\text{weakly in }H(E).
\]
Because $H_+(E)$ is convex and closed, it is weakly closed; thus $g\in H_+(E)$. We decompose the functional 
\begin{equation*}
    I(g)=\underset{\eqqcolon \,h(g)}{\underbrace{\int_E \left(\frac{\nu}{2}a(y)^2(g'(y))^2+\kappa(y)g(y)^2\right)\eta(y)}}+\underset{\eqqcolon f(g)}{\underbrace{\int_E\left(\frac{-2g(y)^{2-\nu}}{2-\nu}\right)\eta(y)}}.
\end{equation*}
Note that $h$ is convex; hence, by arguments similar to those in \citet{GLT:VariationalApproach(2025)}, $h$ is weakly sequentially lower semicontinuous on $H(E)$. Lemma \ref{lem. weak sequential continuity of f} guarantees that $f$ is weakly sequentially continuous. Therefore,
\begin{equation*}
    \liminf_{n \to \infty}I(g_n)\geq \liminf_{n \to \infty}h(g_n)+\lim_{n \to \infty}f(g_n)\geq h(g)+f(g)=I(g).
\end{equation*}
Thus, a minimiser exists. Boundedness and $C^2$ regularity follow directly from Lemma \ref{lem. any minimiser is bounded and C^2}.
\end{proof} 

\begin{proof}[Proof of Lemma \ref{lem. any minimiser is bounded and C^2}]
Let $g \in H(E)_+$ be a solution to \eqref{eq. 1-d minisation problem}. For $\nu \geq 1$ the conclusion is contained in \citet[Theorem 5.2]{GLT:VariationalApproach(2025)}; we therefore focus on the regime $\nu \in (0,1)$, in which the functional $I$ ceases to be convex and uniqueness of the minimiser is unavailable a priori. Steps 4--6 of \citet[Theorem 5.2]{GLT:VariationalApproach(2025)} adapt to the present setting under the correspondence $D \leftrightarrow 1/\nu$ and $\alpha \leftrightarrow 2-\nu$, with the upper-bound step replaced by a direct strict-inequality contradiction with minimality.

For each $y \in E$, define $q_y : (0,\infty) \to \mathbb{R}$ by
\[
q_y(x) := \kappa(y) x^2 - \tfrac{2}{2-\nu}\, x^{2-\nu}.
\]
Then $q_y'(x) = 2 x^{1-\nu}\big(\kappa(y) x^{\nu} - 1\big)$, so $q_y$ has a unique critical point $x^*(y) = \kappa(y)^{-1/\nu}$ and is strictly increasing on $[x^*(y),\infty)$.

\smallskip
\emph{Step 1 (Strict positivity).} If $g$ vanished on a non-empty interval, the perturbation $g + \varepsilon \varphi$ (for a smooth bump function $\varphi$) would change the functional by $A\varepsilon^2 - B\varepsilon^{2-\nu} < 0$ for $A,B>0$ and small $\varepsilon$, contradicting minimality. If $g(y_0) = 0$ for some $y_0 \in E$, the argument of \citet[Theorem 5.2]{GLT:VariationalApproach(2025)}, with $D \leftrightarrow 1/\nu$, $\alpha \leftrightarrow 2-\nu$, gives a similar contradiction. Hence $g > 0$ on $E$.

\medskip
\emph{Step 2 (Upper bound).} Set $\bar{\kappa} := \inf_{y \in E} \kappa(y) > 0$ (Assumption \ref{assum. kappa is strictly positive}) and $M := \bar{\kappa}^{-1/\nu}$. Define $\widetilde{g} := g \wedge M$. Then $\widetilde{g} \in H(E)_+$ since $0 \leq \widetilde{g} \leq g$ and $|\widetilde{g}'| \leq |g'|$ a.e.

Let $A := \{y \in E : g(y) > M\}$. On $A$, $\widetilde{g} = M$ and $\widetilde{g}' = 0$ a.e.; on $A^c$, $\widetilde{g} = g$. Hence
\[
I(g) - I(\widetilde{g}) = \int_A \!\left[\tfrac{\nu a^2}{2}(g')^2 + \big(q_y(g) - q_y(M)\big)\right]\!\eta\,dy.
\]
Since $\kappa(y) \geq \bar{\kappa}$ implies $x^*(y) \leq M$, both $M$ and $g(y)$ lie in $[x^*(y), \infty)$ on $A$, with $g(y) > M$. The strict monotonicity of $q_y$ on $[x^*(y),\infty)$ yields $q_y(g(y)) - q_y(M) > 0$ pointwise on $A$. If $A$ has positive Lebesgue measure, then since $\eta > 0$ on $E$,
\[
I(\widetilde{g}) < I(g),
\]
contradicting minimality. Hence $A$ has Lebesgue measure zero, and continuity of $g$ gives $g \leq M$ on $E$. Combined with Step 1, $0 < g \leq M$ on $E$.

\medskip
\emph{Step 3 (Regularity).} Steps 5--6 from \citet[Theorem 5.2]{GLT:VariationalApproach(2025)} give $g \in C^2(E)$.

\end{proof}

Next, we show that the candidate value function $\hat{V}(x,y)=\frac{x^{1-\gamma}}{1-\gamma}g(y)^m$ solves the BSDE \eqref{eq. Utility Process}. We collect several preparatory lemmas that will be used throughout this appendix. Define
 \begin{align} \label{eq. Dt change of measure}
     D_t\coloneqq\mathcal{E}\left(\int_0^{\cdot}(1-\gamma)\hat{\pi}^{\top} \sigma \bar{\rho}dB_u \right)_{t}\mathcal{E}\left(\int_0^{\cdot}((1-\gamma)\Upsilon^{\top}\hat{\pi}+m\,a^2\frac{g'}{g})^{\top}\frac{1}{a}dW_u\right)_{t},
 \end{align}
 where $dZ_t=\bar{\rho}dB_t+\rho dW_t$, $B$ is an $n$-dimensional Brownian motion independent of $W$ and $\bar{\rho}\bar{\rho}^\top+\rho\rho^\top=I_{n \times n}$.

\begin{lemma} \label{lem. discounted utility of wealth}
    Under the candidate optimal controls $(\hat{\pi},\hat{l})$ (recall \eqref{eq. Candidate optimal controls}), wealth over the interval $[s,t]$ satisfies

    \begin{equation*}
        (X_t^{\hat{\pi},\hat{l}})^{1-\gamma}=(X_s^{\hat{\pi},\hat{l}})^{1-\gamma}\frac{g(Y_s)^m}{g(Y_t)^m}\exp\left(-\int_s^t\theta\, \hat{l}_u\,du\right)\frac{D_t}{D_s}.
    \end{equation*}

\end{lemma}

\begin{proof}[Proof of Lemma \ref{lem. discounted utility of wealth}] We adapt the arguments in \citet[Theorem 3.3]{GW:ConsumptionIncompleteMarkets(2020)} and, for clarity, drop the superscript dependence on the optimal policies. Recall that $X_t$ satisfies \eqref{eq. Wealth Process Dynamics}; thus
 \begin{equation} \label{eq. discounted wealth representation}
     X_t^{1-\gamma}=(X_s)^{1-\gamma}\exp\left((1-\gamma)\int_s^tr+\hat{\pi}^{\top}\mu-\frac{\hat{\pi}^{\top}\Sigma\hat{\pi}}{2}-\hat{l}\,du+(1-\gamma)\int_s^t\hat{\pi}^{\top}\sigma dZ_u\right).
 \end{equation}
 Using the HJB equation \eqref{eq. HJB}, we can rewrite the finite-variation part of the exponent:
 \begin{align} \label{eq. drift of discounted wealth}
     (1-\gamma) \nonumber&\left(r+\hat{\pi}^{\top}\mu-\frac{\hat{\pi}^{\top}\Sigma\hat{\pi}}{2}-\hat{l}\right)\\ \nonumber
     & = -\hat{l}\theta-(1-\gamma)^2\frac{\hat{\pi}^{\top}\Sigma\hat{\pi}}{2}-m(1-\gamma)\hat{\pi}^{\top}\Upsilon\frac{g'}{g}-\frac{a^2m}{2}\left(\frac{g'}{g}\right)^2\\ \nonumber
     & \qquad -m\left(\frac{b\,g'}{g}+\frac{a^2\,g''}{2g}-\frac{a^2}{2}\left(\frac{g'}{g}\right)^2\right).\\
 \end{align}
 Applying It\^o's formula to $\ln g(Y_t)$ yields
 \begin{equation} \label{eq. Ito formula log g}
     \ln g(Y_t)-\ln g(Y_s)=\int_s^t\left(\frac{b\,g'}{g}+\frac{a^2\,g''}{2g}-\frac{a^2}{2}\left(\frac{g'}{g}\right)^2\right)du+\int_s^t\frac{a\,g'}{g}dW_u.
 \end{equation}
 Substituting \eqref{eq. drift of discounted wealth} and \eqref{eq. Ito formula log g} into the exponent of \eqref{eq. discounted wealth representation} yields
 \begin{align*}
     (1-\gamma)& \int_s^t r+\hat{\pi}^{\top}\mu-\frac{\hat{\pi}^{\top}\Sigma\hat{\pi}}{2}-\hat{l}\,du+(1-\gamma)\int_s^t\hat{\pi}^{\top}\sigma dZ_u \\
     & =  -\int_s^t\hat{l}\theta du-\int_s^t(1-\gamma)^2\frac{\hat{\pi}^{\top}\Sigma\hat{\pi}}{2}-m(1-\gamma)\hat{\pi}^{\top}\Upsilon\frac{g'}{g}-\frac{a^2m}{2}\left(\frac{g'}{g}\right)^2du\\
     & \qquad -m\left(\ln g(Y_t)-\ln g(Y_s)\right)+(1-\gamma)\int_s^t\hat{\pi}^{\top}\sigma dZ_u+\int_s^t\frac{m a\,g'}{g}dW_u.
 \end{align*}
 Collecting the remaining terms yields the final representation of the exponent of \eqref{eq. discounted wealth representation}:
 \begin{align} \label{eq. exponent of discounted utility of wealth}
          (1-\gamma)& \int_s^t r+\hat{\pi}^{\top}\mu-\frac{\hat{\pi}^{\top}\Sigma\hat{\pi}}{2}-\hat{l}\,du+(1-\gamma)\int_s^t\hat{\pi}^{\top}\sigma dZ_u \nonumber\\
          &= -\int_s^t\hat{l}\theta du -m\ln\left( g(Y_t)/g(Y_s)\right)+\ln(D_t/D_s),
 \end{align}
Substituting \eqref{eq. exponent of discounted utility of wealth} into \eqref{eq. discounted wealth representation} yields
     \begin{equation*}
        X_t^{1-\gamma}=(X_s)^{1-\gamma}\frac{g(Y_s)^m}{g(Y_t)^m}\exp\left(-\int_s^t\theta\, \hat{l}_u\,du\right)\frac{D_t}{D_s}.
    \end{equation*}
\end{proof}

\begin{lemma} \label{lem. D_t is a Martingale}
Let $\hat{\pi}$ be the candidate optimal investment policy. Then the process $D_t$ is an $(\mathbb{F},\mathbb{P}^y)$ martingale. Furthermore $\tilde{\mathbb{P}}|_{\mathcal{F}_t}=D_t\mathbb{P}^y|_{\mathcal{F}_t}$, where $\tilde{\mathbb{P}}$ is the unique solution to the martingale problem on $\mathbb{R}^{n} \times E$ for
\begin{align} \label{eq. MP for tilde P}
\hat{L} &= \frac{1}{2} \sum_{i,j=1}^{n+1} \hat{A}_{i,j} \frac{\partial^2}{\partial x_i \partial x_j} + \sum_{i=1}^{n+1} \hat{b}_i \frac{\partial}{\partial x_i}\,;\\
\hat{A} & = \begin{pmatrix}
\Sigma & \Upsilon \\
\Upsilon^\top & a^2
\end{pmatrix}; \quad
\hat{b} = \begin{pmatrix} 
\frac{\mu}{\gamma} + \frac{m}{\gamma}\Upsilon \frac{g'}{g}\\[2mm]
b-(1-\frac{1}{\gamma})\Upsilon^{\top}\Sigma^{-1}\mu + a^2 \frac{g'}{g}
\end{pmatrix}, \nonumber
\end{align}
with $\mathbb{P}^y|_{\mathcal{F}_0}=\tilde{\mathbb{P}}|_{\mathcal{F}_0}$.
\end{lemma}
\begin{proof}[Proof of Lemma \ref{lem. D_t is a Martingale}]
    The martingale problem \eqref{eq. MP for tilde P} depends only on the final coordinate. Hence, it is sufficient to prove there exists a unique weak solution to the scalar SDE
    \begin{equation*}
        dY_t=\left(b(Y_t)-(1-\frac{1}{\gamma})\Upsilon^{\top}(Y_t)\Sigma^{-1}(Y_t)\mu(Y_t) + a^2(Y_t) \frac{g'(Y_t)}{g(Y_t)}\right)dt + a(Y_t)dW_t.
    \end{equation*}
    This follows from Assumptions \ref{assum. Market model is well posed} and \ref{assum. Recurrence}, and Lemma \ref{lem. any minimiser is bounded and C^2}. The details are contained in \citet[Theorem 5.2]{GLT:VariationalApproach(2025)}. Hence, the martingale problem \eqref{eq. MP for tilde P} is well posed. Next, $\tilde{\mathbb{P}}|_{\mathcal{F}_t}=D_t\mathbb{P}^y|_{\mathcal{F}_t}$ follows directly from \citet[Lemma A.3]{GW:ConsumptionIncompleteMarkets(2020)}.
\end{proof}
Before proving Proposition \ref{prop. Candidate Value function solves BSDE}, we observe that for all $t\geq0$
\begin{equation} \label{eq. integral of g divergence}
    \int_t^\infty\theta\,g(Y_u)^{-\frac{m}{\delta \theta}}du=\infty \quad \mathbb{P}^y-a.s.
\end{equation}
Lemma \ref{lem. any minimiser is bounded and C^2} yields $0<g\leq M$ and hence $g^{-1}\geq \frac{1}{M}>0$. By assumption, $\frac{m}{\delta \theta}\geq1$; thus $g^{-\frac{m}{\delta \theta}}\geq M^{-\frac{m}{\delta \theta}}>0$, which implies \eqref{eq. integral of g divergence} because $\theta>0$ (Assumption \ref{assum. Parameter regime}). As $\mathbb{P}^y$ and $\tilde{\mathbb{P}}$ are equivalent (Lemma \ref{lem. D_t is a Martingale}), the above integral is divergent $\tilde{\mathbb{P}}$-almost surely.

\begin{proof}[Proof of Proposition \ref{prop. Candidate Value function solves BSDE}]
Throughout we drop the superscript dependence on the candidate optimal policies for ease of notation. 
\begin{align*}
   \frac{(X_t\,\hat{l}_t)^{1-\delta}}{1-\delta}\left((1-\gamma)V^{\hat{l}}_t\right)^{1-\frac{1}{\theta}}&= \frac{1}{1-\delta}g(Y_t)^{-(1-\delta)\frac{m}{\delta \theta}}\,X_t^{1-\delta}(X_t^{1-\gamma}g(Y_t)^m)^{1-\frac{1}{\theta}}\\
    & = \frac{X_t^{1-\gamma}}{1-\delta}g(Y_t)^{-(1-\delta)\frac{m}{\delta \theta}+m(1-\frac{1}{\theta})}.
\end{align*}
Thus, by Lemma \ref{lem. discounted utility of wealth} on an interval $[s,t]$
\begin{align*}
  \frac{(X_t\,\hat{l}_t)^{1-\delta}}{1-\delta}\left((1-\gamma)V^{\hat{l}}_t\right)^{1-\frac{1}{\theta}}=\frac{(X_s)^{1-\gamma}}{1-\gamma}g(Y_s)^m\,\theta g(Y_t)^{-\frac{m}{\delta \theta}}\exp\left(-\int_s^t\theta \,g(Y_u)^{-\frac{m}{\delta \theta}}du\right)\frac{D_t}{D_s}. 
\end{align*}
Next, consider for $t\leq u\leq T$,
\begin{align*}
    \mathbb{E}\Big[\int_t^{T}\frac{(X_u\,\hat{l}_u)^{1-\delta}}{1-\delta} & \left((1-\gamma)V^{\hat{l}}_u\right)^{1-\frac{1}{\theta}}du \big| \mathcal{F}_t\big]\\
    &=\frac{X_t^{1-\gamma}}{1-\gamma}g(Y_t)^m\mathbb{E}\left[\int_t^T\theta \,g(Y_u)^{-\frac{m}{\delta \theta}}\exp\left(-\int_t^u \theta\, g(Y_s)^{-\frac{m}{\delta \theta}}ds\right)\frac{D_u}{D_t}du \big| \mathcal{F}_t\right]\\
    & = \frac{X_t^{1-\gamma}}{1-\gamma}g(Y_t)^m\mathbb{E}^{\tilde{\mathbb{P}}}\left[\int_t^T\theta \,g(Y_u)^{-\frac{m}{\delta \theta}}\exp\left(-\int_t^u \theta\, g(Y_s)^{-\frac{m}{\delta \theta}}ds\right)du \big| \mathcal{F}_t\right]\\
    & = \frac{X_t^{1-\gamma}}{1-\gamma}g(Y_t)^m\left(1-\mathbb{E}^{\tilde{\mathbb{P}}}\left[\exp\left(-\int_t^T \theta\,g(Y_u)^{-\frac{m}{\delta \theta}}du\right)\Big| \mathcal{F}_t\right]\right).
\end{align*}
The second equality follows from Lemma \ref{lem. D_t is a Martingale}. Thus, taking limits as $T\to \infty$ and using monotone convergence, \eqref{eq. integral of g divergence} yields
\begin{align*}
    \mathbb{E}\Big[\int_t^{\infty}\frac{(X_u\,\hat{l}_u)^{1-\delta}}{1-\delta} & \left((1-\gamma)V^{\hat{l}}_u\right)^{1-\frac{1}{\theta}}du \big| \mathcal{F}_t\big]=\frac{X_t^{1-\gamma}}{1-\gamma}g(Y_t)^m.
\end{align*}
But $\frac{X_t^{1-\gamma}}{1-\gamma}g(Y_t)^m=\hat{V}(X_t,Y_t)$, and so $\hat{V}$ solves the BSDE \eqref{eq. Utility Process}.
\end{proof}

Next, we introduce a transversality lemma used in the proof of Lemma \ref{lem. Perturbed candidate value function is supersolution}.

\begin{lemma} \label{lem. Transversality}
    Let $(\hat{\pi},\hat{l})$ be the candidate optimal policies. Then
    \begin{equation*}
        \lim_{t\to \infty}\mathbb{E}[\hat{V}(X_t^{\hat{\pi},\hat{l}},Y_t)]=0.
    \end{equation*}
\end{lemma}

\begin{proof}[Proof of Lemma \ref{lem. Transversality}]
    Recall $\hat{V}(X_t^{\hat{\pi},\hat{l}},Y_t)=\frac{(X_t^{\hat{\pi},\hat{l}})^{1-\gamma}}{1-\gamma}g(Y_t)^m$. Hence, by \eqref{eq. discounted wealth representation},
    \begin{align*}
        \mathbb{E}[\hat{V}(X_t^{\hat{\pi},\hat{l}},Y_t)]
        &=\hat{V}(x,y)\mathbb{E}\left[\exp\left(-\int_0^t\theta\, \hat{l}_u\,du\right)D_t\right]\\
        &= \hat{V}(x,y)\mathbb{E}^{\tilde{\mathbb{P}}}\left[\exp\left(-\int_0^t\theta\, g(Y_u)^{-\frac{m}{\delta\,\theta}}\,du\right)\right].
    \end{align*}
    Thus, taking limits as $t\to \infty$ and using monotone convergence, \eqref{eq. integral of g divergence} yields
\[
        \lim_{t\to \infty}\mathbb{E}[\hat{V}(X_t^{\hat{\pi},\hat{l}},Y_t)]=0.
\]
\end{proof}

\begin{proof}[Proof of Lemma \ref{lem. Perturbed candidate value function is supersolution}]
Consider the perturbed HJB equation for $\hat{V}(z,y)$, where $(z,y)\in \mathbb{R}_+\times E$. As the dynamics of $Z_t^\varepsilon$ in \eqref{eq. Dynamics of Z_t^varepsilon} are equivalent to the wealth process dynamics \eqref{eq. Wealth Process Dynamics}, the HJB equation for $\hat{V}(z,y)$ has the same functional form as the original HJB equation for $\hat{V}(x,y)$, where $(x,y)\in \mathbb{R}_+\times E$. In addition, the supremum is attained for $(\pi,l)=(\hat{\pi},\hat{l})$ because
\begin{align*}
    (\pi_t^z)^{\top}&=\frac{X_t\hat{\pi}_t^{\top}+\varepsilon \hat{X}_t\hat{\pi}_t^{\top}}{Z_t^{\varepsilon}}=\hat{\pi}_t^{\top},\\
    l_t^z&=\frac{X_t\hat{l}_t+\varepsilon \hat{X}_t\hat{l}_t}{Z_t^{\varepsilon}}=\hat{l}_t.
\end{align*}
Thus, for an arbitrary policy we have the inequality
\begin{align} \label{eq. perturbed HJB in equality}
    \mathcal{L}^{\pi,l}\hat{V}(Z_t^\varepsilon,Y_t)+\frac{(C_t^{\varepsilon})^{1-\delta}}{1-\delta}\left((1-\gamma)\hat{V}(Z_t^\varepsilon,Y_t)\right)^{1-\frac{1}{\theta}}\leq0,
\end{align}
where $\mathcal{L}^{\pi,l}$ is the generator for $(Z_t^\varepsilon,Y_t)$. Next, we apply It\^o's Lemma to $\Phi_t^\varepsilon \coloneqq \hat{V}(Z_t^\varepsilon,Y_t)$, which yields
\begin{align*}
    d \Phi_t^\varepsilon =\mathcal{L}^{\pi,l}\Phi_t^\varepsilon dt+dN_t,
\end{align*}
where $N_t$ is the local martingale
\begin{align*}
    dN_t=\partial_z \Phi_t^\varepsilon\,Z_t^\varepsilon (\pi_t^z)^{\top}\sigma dB_t+ \partial_y \Phi_t^\varepsilon\,a \,dW_t.
\end{align*}
Fix arbitrary bounded stopping times $\tau_1\leq \tau_2$ and define for $n \in \mathbb{N}$, $\zeta_n\coloneqq \inf\{s\geq \tau_1: \langle N\rangle_s-\langle N \rangle_{\tau_1} \geq n \}$. Next, we apply It\^o's Lemma once again to obtain
\begin{align*}
    \Phi_{\tau_1}^\varepsilon &= \Phi_{\tau_2 \wedge \zeta_n}^\varepsilon - \int_{\tau_1}^{\tau_2 \wedge \zeta_n}\mathcal{L}^{\pi,l}\Phi_s^\varepsilon ds+N_{\tau_1}-N_{\tau_2 \wedge \zeta_n}\\
    & \geq \Phi_{\tau_2 \wedge \zeta_n}^\varepsilon+\int_{\tau_1}^{\tau_2 \wedge \zeta_n}\frac{(C_s^{\varepsilon})^{1-\delta}}{1-\delta}\left((1-\gamma)\Phi_s^{\varepsilon}\right)^{1-\frac{1}{\theta}}ds+N_{\tau_1}-N_{\tau_2 \wedge \zeta_n}
\end{align*}
Note $(N_{t\wedge\zeta_n} - N_{\tau_1})_{t\geq\tau_1}$ is an 
$L^2$-bounded local martingale and hence a true martingale. 
Since $\tau_1 \leq \tau_2 \wedge \zeta_n$ are bounded stopping times, 
the optional sampling theorem gives $\mathbb{E}[N_{\tau_2\wedge\zeta_n} - N_{\tau_1}\mid\mathcal{F}_{\tau_1}] = 0$. Taking conditional expectations yields
\begin{align*}
    \Phi_{\tau_1}^\varepsilon\geq \mathbb{E}\left[\Phi_{\tau_2 \wedge \zeta_n}^\varepsilon\mid \mathcal{F}_{\tau_1}\right]+\mathbb{E}\left[\int_{\tau_1}^{\tau_2 \wedge \zeta_n}\frac{(C_s^{\varepsilon})^{1-\delta}}{1-\delta}\left((1-\gamma)\Phi_s^{\varepsilon}\right)^{1-\frac{1}{\theta}}ds\mid \mathcal{F}_{\tau_1}\right].
\end{align*}
Since $\hat{V}$ is non-decreasing in its first argument,
\begin{align*}
    \Phi_{\tau_2 \wedge \zeta_n}=\hat{V}(X_{\tau_2 \wedge \zeta_n}^{\pi,l}+\varepsilon\hat{X}_{\tau_2 \wedge \zeta_n},Y_{\tau_2 \wedge \zeta_n})\geq \hat{V}(\varepsilon\hat{X}_{\tau_2 \wedge \zeta_n},Y_{\tau_2 \wedge \zeta_n}) \quad \mathbb{P}^y-a.s.
\end{align*}
If $\gamma \in (0,1)$ then $\hat{V}(\varepsilon\hat{X}_t,Y_t)=\frac{(\varepsilon\hat{X})^{1-\gamma}}{1-\gamma}g(Y_t)^m\geq0$. If $\gamma >1$ then $\hat{V}(\varepsilon\hat{X}_t,Y_t) \leq 0$. Using Lemma \ref{lem. discounted utility of wealth}
\begin{align*}
    \hat{V}(\varepsilon\hat{X}_t,Y_t)&=\hat{V}(\varepsilon,y)\exp\left(-\int_0^t\theta g(Y_u)^{-\frac{m}{\delta \theta}}du\right)D_t\geq \hat{V}(\varepsilon,y)D_t.   
\end{align*}
The lower bound $\hat{V}(\varepsilon\hat{X}_t,Y_t)\geq \hat{V}(\varepsilon,y)D_t$ and uniform integrability of $D_t$ allow conditional Fatou and conditional monotone convergence as $n \to \infty$.
Taking $\liminf$ as $n\to \infty$ yields
\begin{align*} \label{eq. Supersolution for perturbed consumption}
    \Phi_{\tau_1}^\varepsilon\geq \mathbb{E}\left[\Phi_{\tau_2 }^\varepsilon\mid \mathcal{F}_{\tau_1}\right] + \mathbb{E}\left[\int_{\tau_1}^{\tau_2}\frac{(C_s^{\varepsilon})^{1-\delta}}{1-\delta}\left((1-\gamma)\Phi_s^{\varepsilon}\right)^{1-\frac{1}{\theta}}ds\mid \mathcal{F}_{\tau_1}\right].
\end{align*}
Furthermore, $\liminf_{t \to \infty}\mathbb{E}[\Phi_t^\varepsilon]\geq \liminf_{t \to \infty}\mathbb{E}[\hat{V}(\varepsilon \hat{X}_t,Y_t)]=0$ by Lemma \ref{lem. Transversality}. Hence, for an arbitrary policy pair $(\pi,l)\in \mathcal{A}$, $\Phi_t^\varepsilon=\hat{V}(Z_t^\varepsilon,Y_t)$ is a supersolution to \eqref{eq. Utility Process}.
\end{proof}

\begin{proof}[Proof of Lemma \ref{lem. Comparison lemma}]
    As $\gamma >1$, $V,\hat{V} \in [-\infty,0]$. For all bounded stopping times $\tau_1\leq \tau_2$ we have
    \begin{equation} \label{eq. supersolution inequality}
\hat{V}(Z_{\tau_1}^\varepsilon,Y_{\tau_1})\geq \mathbb{E}\left[\hat{V}(Z_{\tau_2}^\varepsilon,Y_{\tau_2})+ \int_{\tau_1}^{\tau_2}f(C_s, \hat{V}(Z_{s}^\varepsilon,Y_s))ds \mid\mathcal{F}_{\tau_1}\right],
    \end{equation}
    \begin{equation} \label{eq. subsolution inequality}
            -V_{\tau_1}\geq -\mathbb{E}\left[V_{\tau_2}+ \int_{\tau_1}^{\tau_2}f(C_s,V_s)ds \mid\mathcal{F}_{\tau_1}\right].        
    \end{equation}

    We proceed by contradiction. Suppose that there exist a finite stopping time $\tau$ and a set $A \in \mathcal{F}_\tau$ of positive measure such that $\hat{V}_{\tau}(\omega)<V_{\tau}(\omega)$. Then $\mathbb{E}[\textbf{1}_{A}(\hat{V}_{\tau}-V_{\tau})]<0$. Introduce the stopping time
    \begin{equation*}
        \sigma \coloneqq \inf\{s\geq \tau : V_{s+}-\hat{V}_{s+}\leq 0 \}.
    \end{equation*}
     By construction, for all $\omega \in A$ and $s \in [\tau(\omega),\sigma(\omega))$
        \begin{align*}
        0\geq V_s(\omega)\geq \hat{V}_s(\omega)\geq \hat{V}(\varepsilon,y)\,D_s(\omega).
    \end{align*}
    As $D_t$ is uniformly integrable and finite, this necessarily implies $V$ is finite with finite expectation for $\omega \in A$ and $s \in [\tau(\omega),\sigma(\omega))$. Fix $n \in \mathbb{N}$, and define $\tau_1\coloneqq \tau \wedge n$ and $\tau_2 \coloneqq \sigma \wedge n$. Multiply \eqref{eq. supersolution inequality} and \eqref{eq. subsolution inequality} by $\textbf{1}_{A \cap \{\tau\leq n\}}$ and take expectations, yielding
    \begin{align} \label{eq. supersolution inequality with indicators}
\mathbb{E}[\textbf{1}_{A \cap \{\tau\leq n\}}\hat{V}(Z_{\tau \wedge n}^\varepsilon,Y_{\tau \wedge n})]
&\geq \mathbb{E}\left[\textbf{1}_{A \cap \{\tau\leq n\}}\hat{V}(Z_{\sigma \wedge n}^\varepsilon,Y_{\sigma \wedge n})\right]\nonumber\\
&\quad + \mathbb{E}\left[\textbf{1}_{A \cap \{\tau\leq n\}}\int_{\tau \wedge n}^{\sigma \wedge n}f(C_s, \hat{V}(Z_{s}^\varepsilon,Y_s))ds \right],
    \end{align}
    \begin{align} \label{eq. subsolution inequality with indicators}
\mathbb{E}[-\textbf{1}_{A \cap \{\tau\leq n\}}V_{\tau \wedge n}]
&\geq -\mathbb{E}\left[\textbf{1}_{A \cap \{\tau\leq n\}}V_{\sigma \wedge n}\right]\nonumber\\
&\quad - \mathbb{E}\left[\textbf{1}_{A \cap \{\tau\leq n\}}\int_{\tau \wedge n}^{\sigma \wedge n}f(C_s,V_s)ds\right]\geq 0.
    \end{align}
    Since $\mathbb{E}[-\textbf{1}_{A \cap \{\tau\leq n\}}V_{\tau_1}]\geq0$ is finite, the expectations on the right-hand side of \eqref{eq. subsolution inequality with indicators} are also finite. Thus, we may add \eqref{eq. subsolution inequality with indicators} and \eqref{eq. supersolution inequality with indicators}:

    \begin{align*}
        \mathbb{E}&\left[\textbf{1}_{A \cap \{\tau\leq n\}}\left(\hat{V}(Z_{\tau \wedge n}^\varepsilon,Y_{\tau \wedge n})-V_{ \tau \wedge n}\right)\right]\\
        & \geq \mathbb{E}\left[\textbf{1}_{A \cap \{\tau\leq n\}}\left(\hat{V}(Z_{\sigma \wedge n}^\varepsilon,Y_{\sigma \wedge n})-V_{ \sigma \wedge n}\right)\right]\\
         & \quad+ \mathbb{E}\left[\textbf{1}_{A \cap \{\tau\leq n\}}\int_{\tau \wedge n}^{\sigma \wedge n}f(C_s, \hat{V}(Z_{s}^\varepsilon,Y_s))-f(C_s,V_s)ds \right],
    \end{align*}
    where all integrals are well-defined. Since $f$ is decreasing in its second argument for $\theta\in(0,1)$, the integral term is non-negative on $A\cap\{\tau\leq n\}$. It remains to pass to the limit in the terminal term. On $\{\sigma<n\}$, the definition of $\sigma$ gives
    \[
        \hat{V}(Z_{\sigma}^\varepsilon,Y_{\sigma})-V_{\sigma}\geq0.
    \]
    On $\{\tau\leq n\leq\sigma\}$, we have $\hat{V}(Z_n^\varepsilon,Y_n)-V_n\geq \hat{V}(Z_n^\varepsilon,Y_n)\geq \hat{V}(\varepsilon\hat{X}_n,Y_n)$, so Lemma \ref{lem. Transversality} gives
    \[
        \liminf_{n\to\infty}\mathbb{E}\!\left[\textbf{1}_{A\cap\{\tau\leq n\leq\sigma\}}\left(\hat{V}(Z_n^\varepsilon,Y_n)-V_n\right)\right]\geq0.
    \]
    Finally, on the left-hand side,
    \[
        \textbf{1}_{A\cap\{\tau\leq n\}}\left|\hat{V}(Z_{\tau\wedge n}^\varepsilon,Y_{\tau\wedge n})-V_{\tau\wedge n}\right|
        \leq \textbf{1}_A |\hat{V}(\varepsilon,y)|D_\tau,
    \]
    and the right-hand side is integrable by uniform integrability of $D$. Dominated convergence therefore yields
    \begin{equation*}
        \mathbb{E}\left[\textbf{1}_{A}\left(\hat{V}(Z_{\tau}^\varepsilon,Y_{\tau})-V_{ \tau }\right)\right]\geq0.
    \end{equation*}
    This contradicts the choice of $A$.

\end{proof}

\begin{proof}[Proof of Corollary \ref{cor. solution to the minimisation problem is unique}]
Let $g_1$ and $g_2$ be solutions to the variational problem \eqref{eq. 1-d minisation problem}. Theorem \ref{thm. Main Theorem} yields
\[
\frac{x^{1-\gamma}}{1-\gamma}g_1(y)^m=\sup_{(\pi,l) \in \mathcal{A}}V_0^{\pi,l}=\frac{x^{1-\gamma}}{1-\gamma}g_2(y)^m.
\]
As $m\neq0$, the mapping $t \mapsto t^m$ is bijective on $(0,\infty)$. Thus $g_1(y)=g_2(y)$ for all $y \in E$.
\end{proof}
\nihil{
\begin{proof}[Proof of Corollary \ref{cor. Impossibility theorem}]
We follow the proof of Proposition \ref{prop. Candidate Value function solves BSDE} verbatim to obtain
\begin{align*}
  \mathbb{E}\Big[\int_t^{\infty}\frac{(X_u\,\hat{l}_u)^{1-\delta}}{1-\delta} & \left((1-\gamma)V^{\hat{l}}_u\right)^{1-\frac{1}{\theta}}du \big| \mathcal{F}_t\big]\\
  &=\frac{X_t^{1-\gamma}}{1-\gamma}g(Y_t)^m\left(1-\mathbb{E}^{\tilde{\mathbb{P}}}\left[\exp\left(-\int_t^\infty \theta\,g(Y_u)^{-\frac{m}{\delta \theta}}du\right)\Big| \mathcal{F}_t\right]\right).
\end{align*}
Thus, $V$ is a solution to \eqref{eq. Utility Process} if and only if
\begin{equation*}
    -\int_t^\infty \theta\, g(Y_u)^{-\frac{m}{\delta \theta}}du=-\infty.
\end{equation*}
As $0<g<M$, this is true if and only if $\theta>0$. Hence, if $\theta<0$, $V$ does not solve \eqref{eq. Utility Process} and cannot be the value function.
\end{proof}
}
\bibliographystyle{plainnat}
\bibliography{Ref}

\end{document}